\documentclass[11pt,letterpaper]{article}
\usepackage[T1]{fontenc}
\usepackage[utf8]{inputenc}
\usepackage{lmodern,microtype}
\usepackage[margin=1in]{geometry}
\usepackage{amsmath,amssymb,amsthm,mathtools}
\usepackage{booktabs,array,enumitem,algorithm,aliascnt,xcolor,listings}
\newcommand{\revise}[2]{#2}
\usepackage[backend=biber,style=alphabetic,sorting=nyt,sortcites=true,
  maxbibnames=99,maxalphanames=4,giveninits=true,doi=true,url=false,
  isbn=false,eprint=true]{biblatex}
\usepackage[hidelinks]{hyperref}
\usepackage[nameinlink,noabbrev,capitalise]{cleveref}
\allowdisplaybreaks[2]
\setlist{nosep}
\numberwithin{equation}{section}
\newtheorem{theorem}{Theorem}[section]
\newaliascnt{lemma}{theorem}
\newtheorem{lemma}[lemma]{Lemma}
\aliascntresetthe{lemma}
\newaliascnt{proposition}{theorem}
\newtheorem{proposition}[proposition]{Proposition}
\aliascntresetthe{proposition}
\newaliascnt{corollary}{theorem}

\aliascntresetthe{corollary}
\theoremstyle{definition}
\newaliascnt{definition}{theorem}

\aliascntresetthe{definition}
\theoremstyle{remark}
\newaliascnt{remark}{theorem}

\aliascntresetthe{remark}
\crefname{theorem}{Theorem}{Theorems}
\crefname{lemma}{Lemma}{Lemmas}
\crefname{proposition}{Proposition}{Propositions}
\crefname{corollary}{Corollary}{Corollaries}
\crefname{definition}{Definition}{Definitions}
\crefname{remark}{Remark}{Remarks}
\newcommand{\E}{\mathbb E}
\newcommand{\Prb}{\mathbb P}
\newcommand{\R}{\mathbb R}
\newcommand{\ind}{\mathbf 1}
\newcommand{\dd}{\,\mathrm d}
\newcommand{\SB}{\mathrm{SB}}
\newcommand{\FB}{\mathrm{FB}}
\newcommand{\fl}{\operatorname{fl}}

\newcommand{\rhostar}{\rho_{\mathrm U}}
\DeclareMathOperator{\diag}{diag}

\hypersetup{pdftitle={Sharp Second-Best Welfare in Bilateral and Matching Markets},
  pdfauthor={Zhengyang Liu, Ying Qin, Zihe Wang},pdfsubject={Welfare, Bayesian mechanism design, bilateral trade}}
\title{Sharp Second-Best Welfare in\\ Bilateral and Matching Markets}
\author{%
Zhengyang Liu\thanks{Beijing Institute of Technology. Email: \texttt{zhengyang@bit.edu.cn}}
\and
Ying Qin\thanks{Renmin University of China. Email: \texttt{qinying0420@ruc.edu.cn}}
\and
Zihe Wang\thanks{Renmin University of China. Email: \texttt{wang.zihe@ruc.edu.cn}}
}
\date{}

\begin{document}
\maketitle
\begin{abstract}
How much social welfare must a market lose because values are private?
We determine the sharp ratio of second-best to first-best welfare under
independent nonnegative types, Bayesian incentive compatibility, interim
individual rationality, and no expected budget deficit.
The ratio is approximately $0.8882516903$ for arbitrary priors and
$0.9113893681$ when buyers have monotone hazard rates and sellers are
unrestricted. The arbitrary-prior guarantee holds for every downward-closed
family of feasible matchings; the MHR guarantee holds when every matching
of a compatibility graph is feasible. Both constants are \revise{attained}{sharp} already
in bilateral trade.
The proof keeps the sellers' initial endowment inside the budget
Lagrangian. For arbitrary priors, a common threshold decomposition reduces
the problem to three-parameter power-law distributions. For monotone
hazards, common transformations of buyer and seller scores reduce it to
shifted capped exponential buyers; a typewise allocation certificate and
an explicit worst-case seller satisfy the same boundary equation.
The bilateral-to-matching framework, with a welfare endowment charge
and an MHR-preserving packing argument, transfers these affine inequalities
without loss. We give exact variational characterizations of both
constants and reproducible interval certificates for their numerical
evaluation. Unrestricted signed transfers also permit pointwise strong
budget balance.
\end{abstract}

\clearpage
\setcounter{tocdepth}{2}
\begingroup
\small
\tableofcontents
\endgroup
\clearpage

\section{Introduction}\label{sec:introduction}
A seller initially owns an item, and a buyer may value it more.
If the seller's value is $S$ and the buyer's value is $B$, efficient trade
creates gains $(B-S)_+$, where $(z)_+:=\max\{z,0\}$, and total welfare $\max\{S,B\}$.
With private information, the mechanism must elicit these values and
finance the transfer. The Myerson--Satterthwaite theorem shows that
incentive compatibility, individual rationality, and budget balance can
prevent efficient trade even when the two values are independent~\cite{MS83}.
The \emph{second-best} benchmark asks for the highest expected welfare
compatible with these constraints.

Recent work determines the sharp second-best ratio for gains from trade:
it is $1/2$ in bilateral trade~\cite{LQRW26} and in matching
markets~\cite{BLWZ26}. Distributional restrictions can improve this ratio,
and these improvements can also survive competition between
trades~\cite{LQW26}. Total welfare presents a different extremal problem.
Write $A$ for the expected value of the sellers' initial endowment and $G$
for the first-best gains. The first-best and second-best welfare values are
\[
 W_{\FB}=A+G,\qquad W_{\SB}=A+G_{\SB}.
\]
The same mechanism optimizes welfare and gains on a fixed instance.
Nevertheless, their worst-case approximation ratios can have different
extremizers: the endowment is part of the welfare benchmark, and it is
correlated with how much trade must be sacrificed to balance the budget.
A sharp gains bound alone does not identify that relation.

We resolve the welfare problem for arbitrary independent priors and for
buyers with monotone hazard rates (MHR). We also show that neither sharp
ratio decreases when bilateral trade is replaced by a matching market.
The common principle is to retain the endowment throughout the dual
argument. This leads to affine inequalities whose lower certificates and
worst-case distributions can be derived together.

\subsection{Main Results}
A matching market has finitely many buyers and sellers and a fixed
bipartite compatibility graph. Each seller owns one item, each buyer wants
at most one compatible item, and a downward-closed family $\mathcal F$
specifies the feasible matchings. In an ordinary matching market,
$\mathcal F$ contains every matching of the graph. A buyer has one scalar
value for every compatible item. Types are independent and nonnegative.
The designer knows their distributions.

Throughout, a feasible mechanism is Bayesian incentive compatible (BIC),
interim individually rational (IR), and ex ante weakly budget balanced
(WBB). Transfers may be signed and may occur without trade. These
conventions matter: we do not require dominant-strategy incentive
compatibility or ex post individual rationality.

\begin{theorem}[Sharp welfare ratios]\label{thm:main}
For every such matching market with positive finite expected first-best
welfare,
\[
 W_{\SB}\ge \rho_{\mathrm U}W_{\FB}.
\]
If every buyer is MHR and the market is ordinary, then
\[
 W_{\SB}\ge \rho_{\mathrm M}W_{\FB},
\]
with no restriction on seller distributions. The constants are sharp
already for a single buyer and seller, and satisfy
\begin{align}
 .8882516903&<\rho_{\mathrm U}<.8882516904,
 \label{eq:general-bracket}\\
 .911389368124&<\rho_{\mathrm M}<.911389368127.
 \label{eq:sharp-bracket}
\end{align}
Both ratios are attained by bounded bilateral priors when an MHR buyer
may have an atom at its upper endpoint. Requiring MHR buyers to be
atomless leaves the sharp infimum unchanged. The same guarantees and
sharp constants hold under pointwise strong budget balance.
\end{theorem}

The constants are defined exactly by the variational problems in
\eqref{u:eq:rho-star-definition} and \eqref{eq:w-main}; their decimal
values are certified evaluations. Sharpness is over all instances in
each class, rather than for each fixed graph or distribution.
The unrestricted guarantee permits additional downward-closed feasibility
constraints. The MHR guarantee is stated for ordinary matching markets.
Both require scalar buyer values; item-dependent value vectors are
outside the model.

\begin{table}[t]
\centering
\begin{tabular}{lcc}
\toprule
Priors & Bilateral trade & Matching markets\\
\midrule
Arbitrary independent priors & $\rho_{\mathrm U}\simeq .8882516903$
 & $\rho_{\mathrm U}$\\
MHR buyers; arbitrary sellers & $\rho_{\mathrm M}\simeq .9113893681$
 & $\rho_{\mathrm M}$\\
\bottomrule
\end{tabular}
\caption{Sharp second-best/first-best welfare ratios. All entries use
BIC, interim IR, and WBB; unrestricted signed transfers permit pointwise
SBB as well. The unrestricted matching entry allows any downward-closed
feasibility family; the MHR entry concerns ordinary matching markets.}
\label{tab:results}
\end{table}

Write $\mathcal U$ for unrestricted priors and $\mathcal M$ for MHR
buyers with arbitrary sellers. The proof supplies a stronger intermediate
statement. Let $\Lambda(x)$ be
\revise{the largest expected payment surplus compatible with the allocation $x$,
BIC, and interim IR, and let}{the maximum expected payment surplus compatible with the allocation $x$,
BIC, and interim IR, where payment surplus is expected buyer payments minus expected seller receipts. For a budget-surplus Lagrange multiplier $\alpha\ge0$, let}
\[
 V_{\revise{\lambda}{\alpha}}=\sup_x\{G(x)+\revise{\lambda}{\alpha}\Lambda(x)\}.
\]
Budget duality gives $G_{\SB}=\inf_{\revise{\lambda}{\alpha}\ge0}V_{\revise{\lambda}{\alpha}}$ in the bounded
classes used by the proof. Thus a welfare guarantee of $\rho$ is exactly
the family of inequalities
\begin{equation}\label{eq:intro-affine}
 V_{\revise{\lambda}{\alpha}}+(1-\rho)A\ge\rho G
 \quad\text{for every }\revise{\lambda}{\alpha}\ge0.
\end{equation}
We reduce these inequalities without losing either the endowment term or
the sharp constant.

\begin{theorem}[Welfare-preserving affine reductions]\label{thm:affine}
Fix $\revise{\lambda}{\alpha}>0$, let $a=\revise{\lambda}{\alpha}/(1+\revise{\lambda}{\alpha})$ and $k=1-a$, and fix
$\beta,\kappa\ge0$. The inequality
\begin{equation}\label{eq:affine}
 V_{\revise{\lambda}{\alpha}}+\kappa A\ge\beta G
\end{equation}
holds for all bounded bilateral instances in $\mathcal U$ if and only if
it holds for all bounded downward-closed matching markets in $\mathcal U$.
For $\mathcal M$, the same equivalence holds with ordinary matching
markets. Moreover, bilateral verification reduces as follows.
\begin{enumerate}[label=(\roman*)]
\item For arbitrary priors, it is equivalent to
$pr/k+\kappa A_a(p)\ge\beta I_a(p,r)$ for every $p,r\in(0,1)$,
where $A_a,I_a$ are defined in \eqref{u:eq:A}--\eqref{u:eq:I}. In the normalized power-law layer introduced below, $p$ is the seller's atom at $S=0$ and $r$ is the buyer's atom at $B=1$.
\item For MHR buyers, it suffices, and is necessary, to consider arbitrary
bounded sellers and positive scalings of
\begin{equation}\label{eq:canonical}
 B=\ell+\min\{E,U-\ell\},\qquad E\sim\operatorname{Exp}(1),
 \qquad 0\le\ell\le a\le U,
\end{equation}
Here $\ell$ and $U$ are the buyer's lower and upper support endpoints, respectively. including boundary limits.
\end{enumerate}
\end{theorem}

The unrestricted extremizer consists of two power-law distributions
with endpoint atoms. Its second-best mechanism has a simple boundary
lottery. The MHR extremizer is a shifted capped exponential buyer paired
with an explicitly constructed seller. In each case, the same dual
expression proves a universal lower bound and supplies the upper
instance.

\subsection{Technical Overview}
We use $\E$ and $\Prb$ for expectation and probability, respectively, and $\ind_A$ for the indicator of an event $A$.
\paragraph{Score reconstruction.}
At a fixed budget multiplier, ironing expresses the optimized Lagrangian
as $k^{-1}\E(u-v)_+$ for monotone buyer and seller scores $u,v$, where $u$ and $v$ denote the ironed normalized buyer score and seller score, respectively.
Reconstructing values from these scores through a positive linear
operator preserves the dual value, decreases the seller mean, and
increases efficient gains. A common threshold decomposition then splits
the two scores simultaneously. The seller mean and dual value add
exactly across layers, whereas efficient gains are subadditive. Every
layer is a power-law pair described by $(a,p,r)$. This proves an exact
three-dimensional characterization of the unrestricted constant.

\paragraph{MHR decomposition.}
The power-law decomposition does not preserve MHR. Instead write an MHR
buyer as $B=g(E)$ with $g$ increasing and concave. Its curvature measure
decomposes it into shifted capped exponential components. Decomposing
only the buyer gives the wrong direction for the positive-part dual
functional. We therefore apply the same monotone score transformation
to the seller in every component. This restores exact additivity of
the dual and the seller mean, while preserving the needed inequality
for efficient gains.

\paragraph{Typewise certificates.}
For a \revise{canonical}{normalized capped-exponential} MHR buyer, serve the highest buyer types with probability
$p(s)$ at seller cost $s$. Requiring the affine welfare inequality to be
tight at every served seller type gives a differential equation.
The endpoint at which service vanishes satisfies
$(1-\rho)L=\rho\E(B-L)_+$. The allocation starts by serving a continuous
buyer tail and ends by randomizing within the buyer's top atom.
One residual equation joins these two parts.
An explicit seller distribution makes this rule maximize the
Lagrangian, so a zero of the residual proves both directions of the
seller optimization. Optimizing the remaining three buyer parameters
gives $\rho_{\mathrm M}$.

\paragraph{Endowment charges.}
For unrestricted priors, we use the cap-monotone local rules and
edge-by-edge composition of~\cite{LQW26}. Charging a seller's full prior
mean in each nonempty first-best cap costs at most $A$: with all other
reports fixed, that seller has at most one candidate partner. This
extends the affine inequality to downward-closed matching markets.
The automatic local-rule theorem requires every buyer probability
vector, so an MHR-only bound does not meet its premise. For ordinary
matching, we instead condition on upper-buyer/lower-seller rectangles.
Upper tails preserve MHR, and rectangle packing charges each endowment
at most once.

\paragraph{Global certification.}
The analytic reductions are exact and account for all priors and budget
multipliers. Evaluating their minima is computer assisted.
For each three-parameter problem we first exclude boundary sequences,
then rule out nonlocal minima by a finite interval cover, and finally
verify strict convexity and opposite signs near the unique \revise{canonical}{normalized}
minimizer. Thus the numerical step proves global sharpness, not just
the accuracy of a candidate. The appendices give all acceptance
inequalities; the accompanying certificate repository supplies the finite
data and replay programs.

\subsection{Related Work}
Fixed-price welfare approximation builds on the connection
between dominant-strategy trading and posted prices~\cite{HR87} and
approximate-efficiency mechanisms~\cite{BD21}. Kang, Pernice, and
Vondr\'ak~\cite{KPV22}, Cai and Wu~\cite{CW23}, and Liu, Ren, and
Wang~\cite{LRW23} sharpen the fixed-price guarantees and impossibility
bounds. Liu, Ren, and Wang show that knowing only the buyer distribution
suffices to \revise{attain}{achieve} the same worst-case ratio as knowing both priors,
and place that ratio in $[.71,.7381]$. Giambartolomei and de
Keijzer~\cite{GK26} narrow the interval to $[.7292,.73805]$; Jiang,
Gao, and Cai~\cite{JGC26} subsequently characterize the ratio exactly,
with value approximately $.7380243357$. Beyond fixed prices, Dobzinski
and Shaulker~\cite{DS26} obtain a $.746$ welfare guarantee using a
buyer-offering mechanism with a reserve. Our benchmark instead
optimizes over all BIC, interim-IR mechanisms with no expected deficit
and unrestricted signed transfers. Its larger feasible class permits
larger welfare ratios. Technically, the typewise certificates and
endpoint conditions of~\cite{LRW23} anticipate part of our analysis;
here the optimized budget Lagrangian also supplies a seller witnessing
sharpness against every mechanism in our class.

The amount and structure of prior information lead to different welfare
problems. D\"utting et al.~\cite{DFL21} obtain welfare approximations
in two-sided markets using one sample from each seller distribution.
Building on~\cite{DFL21,KPV22}, Liu, Ren, and Wang~\cite{LRW23} show
that randomization does not improve the optimal single-seller-sample
welfare ratios: $1/2$ for arbitrary independent priors and $3/4$ for
identical priors. Deng et al.~\cite{DMSWW25} study delegated pricing
when agents learn from samples and obtain constant GFT guarantees for
classes of sample-based pricing rules. Dobzinski and Shaulker~\cite{DS24}
study correlated private values, while Dobzinski et al.~\cite{DEGST25}
study interdependent values. Our theorems assume known, independent
private-value distributions; these information restrictions and
extensions are outside their scope.

Gains from trade provide the budget-dual starting point and the route
to matching markets. Early distribution-dependent guarantees appear
in~\cite{McA08,BM16}. Brustle et al.~\cite{BCWZ17} obtain half of
\emph{second-best} GFT in two-sided markets using duality. For arbitrary
bilateral priors, Deng et al.~\cite{DMSW22} obtain a constant fraction
of \emph{first-best} GFT; Fei~\cite{Fei22} improves this guarantee and
determines the sharp performance of seller pricing for MHR buyers.
Liu et al.~\cite{LQRW26} establish the sharp bilateral second-best/first-best
GFT ratio of $1/2$. In larger markets, welfare approximation extends
from double auctions~\cite{CBKLT16} to combinatorial
auctions~\cite{CBGKRT20}, and Cai et al.~\cite{CGMZ21} study
multidimensional GFT. For scalar-value matching markets with
downward-closed feasibility, Babaioff et al.~\cite{BRTW26} obtain a
constant fraction of first-best GFT, and Bei et al.~\cite{BLWZ26}
prove the sharp second-best/first-best half guarantee. Liu, Qin, and
Wang~\cite{LQW26} give a general bilateral-to-matching reduction and
sharp gains guarantees under MHR and finite-support restrictions.
We use their composition, cap-monotonicity, and seller regularization
directly. Welfare additionally requires controlling the initial
endowment: our full-prior charge gives the unrestricted extension under
downward-closed feasibility, while MHR-preserving rectangle conditioning
gives the ordinary-matching extension. The same distinction drives the
bilateral analysis, where a simultaneous decomposition must preserve
both the endowment and the optimized dual value.

\paragraph{Organization.}
\Cref{sec:model} formulates the model and budget dual.
\Cref{sec:unrestricted,sec:mhr-start} prove the two bilateral reductions
and construct their upper instances.
\Cref{sec:matching} transfers the inequalities to matching markets.
\Cref{sec:extension} handles general priors and payment conventions.
\Cref{app:unrestricted,app:mhr} prove the two global evaluations;
\cref{sec:reproduction} gives the certificate data and replay commands.

\section{Preliminaries}\label{sec:model}
We first fix the feasible allocations and payment conventions. The
budget dual then expresses both welfare problems through the same
affine inequality, which the next two sections analyze bilaterally.

Let $\mathcal E$ be a fixed bipartite compatibility graph.
Seller $j$ initially owns one item and values it at $S_j$; buyer $i$
values any compatible item at $B_i$ and wants at most one.
All types are independent nonnegative Borel random variables. For each buyer $i$ and seller $j$, let $F_{B_i}$ and $F_{S_j}$ denote the CDFs of $B_i$ and $S_j$; whenever a density or probability mass function exists, denote it by $f_{B_i}$ or $f_{S_j}$. In bilateral trade we suppress the agent index and write $F_B,F_S,f_B,f_S$. Let
$\mathcal F$ be a downward-closed family of matchings of $\mathcal E$;
discard edges that occur in no feasible matching. An allocation is a
measurable lottery over $\mathcal F$. The ordinary case takes
$\mathcal F$ to be all matchings of $\mathcal E$.
For an allocation $x$, let $M_x$ denote the random feasible matching induced by $x$ at the realized type profile, and write
\begin{equation}\label{eq:baseline}
 A=\sum_j\E S_j,\qquad
 G(x)=\E\sum_{(i,j)\in M_x}(B_i-S_j),\qquad
 G=\E\max_{M\in\mathcal F}\sum_{(i,j)\in M}(B_i-S_j).
\end{equation}
Thus $A$ is the expected value of the sellers' initial endowment, $G(x)$ is the expected gains from trade generated by allocation $x$, and $G$ is the first-best expected gains from trade. The first-best welfare is $A+G$. We consider instances with
$0<A+G<\infty$.

Let $\revise{X_i(b_i)}{x_{B_i}(b_i)}$ and $\revise{Y_j(s_j)}{x_{S_j}(s_j)}$ be interim service probabilities, and let
$\revise{P_i(b_i)}{P_{B_i}(b_i)}$ and $\revise{P_j(s_j)}{P_{S_j}(s_j)}$ be the buyer's interim payment and seller's
interim receipt. Utilities relative to retaining the initial endowment
are $b_i\revise{X_i}{x_{B_i}}-\revise{P_i}{P_{B_i}}$ and $\revise{P_j}{P_{S_j}}-s_j\revise{Y_j}{x_{S_j}}$. BIC requires truthfulness to maximize
these interim utilities for every type. Interim IR requires their
nonnegativity. WBB requires total expected buyer payments to cover
total expected seller receipts. All transfers are integrable; their signs
and payments at no-trade profiles are unrestricted.
Define $G_{\SB}$ as the supremum of $G(x)$ over such mechanisms.
Then $W_{\SB}=A+G_{\SB}$.

The unrestricted prior class is denoted $\mathcal U$. In the class
$\mathcal M$, each buyer is MHR and sellers are arbitrary.
Precisely, we use the endpoint-inclusive representation
\begin{equation}\label{eq:mhr}
 B=g(E),\qquad E\sim\operatorname{Exp}(1),\qquad
 g\ge0\text{ nondecreasing and concave},\quad g(0)=g(0+).
\end{equation}
It allows deterministic buyers and an upper endpoint atom.
For a continuous nondegenerate distribution it is the usual MHR
condition: at $x=g(t)$, the hazard is $1/g'(t)$ almost everywhere.
Equivalently its survival function is log-concave, with the endpoint
convention above. \Cref{sec:extension} shows that excluding the upper
atom does not change the sharp infimum.

For allocations with monotone interims, let $\Lambda(x)$ be \revise{the largest
attainable expected payment surplus under BIC and interim IR. On the
bounded classes considered first, set}{the maximum expected payment surplus under BIC and interim IR, where payment surplus means expected buyer payments minus expected seller receipts. On the
bounded classes considered first, let $\alpha\ge0$ denote the Lagrange multiplier on the budget-surplus constraint and set}
\begin{equation}\label{eq:dual-def}
 V_{\revise{\lambda}{\alpha}}=\sup_x\{G(x)+\revise{\lambda}{\alpha}\Lambda(x)\},\qquad
 a=\frac{\revise{\lambda}{\alpha}}{1+\revise{\lambda}{\alpha}},\qquad k=1-a.
\end{equation}
For $\revise{\lambda}{\alpha}>0$, $0<a<1$; at $\revise{\lambda}{\alpha}=0$, $V_0=G$.
We recall the payment identities of~\cite{LQW26} and give the duality
extension needed for mixed continuous and finite priors. In particular,
for $0\le\rho\le1$,
\begin{equation}\label{eq:cone-ray}
 W_{\SB}\ge\rho W_{\FB}
 \quad\Longleftrightarrow\quad
 V_{\revise{\lambda}{\alpha}}+(1-\rho)A\ge\rho G
 \quad\text{for all }\revise{\lambda}{\alpha}\ge0.
\end{equation}
The forward direction follows from weak duality; the reverse direction
takes the infimum over multipliers. A single multiplier certifies an
upper bound on second best. A universal lower guarantee needs every
multiplier.

\subsection{Envelope Payments}\label{u:sec:payments}
We use the finite payment characterization and the signed-transfer
budget conversion of \cite[Lemma~2.1 and Appendix~A.1]{LQW26}.
For reference, we record their consequences in the bilateral notation
used by our reductions. Neither changes the feasible allocations or
their welfare.

\revise{Take buyer values $b_1>\cdots>b_n\ge0$ with positive probabilities $w_i$, and seller values $0\le s_1<\cdots<s_m$ with positive probabilities $z_j$. Define cumulative probabilities
$q_i=\sum_{h\le i}w_h$ and $Q_j=\sum_{h\le j}z_h$, with $q_0=Q_0=0$.}{Take buyer values $0\le b_1<\cdots<b_n$ and seller values $0\le s_1<\cdots<s_m$. Define the probability mass functions $f_B(i):=\Prb(B=b_i)>0$ and $f_S(j):=\Prb(S=s_j)>0$, and their cumulative distribution functions $F_B(i):=\sum_{h\le i}f_B(h)$ and $F_S(j):=\sum_{h\le j}f_S(h)$, with $F_B(0)=F_S(0)=0$.}
Let $\revise{X_i}{x_B(b_i)}=\E_S x(S,b_i)$ and $\revise{Y_j}{x_S(s_j)}=\E_B x(s_j,B)$. BIC requires \revise{both sequences to be nonincreasing in the displayed index}{the sequence $x_B(b_i)$ to be nondecreasing in $i$ and $x_S(s_j)$ to be nonincreasing in $j$}.

\begin{lemma}[Finite payment identities]\label{u:lem:payments}
An allocation admits a BIC, interim-IR, ex ante WBB implementation if and only if \revise{$X$ and $Y$ are nonincreasing}{$\{x_B(b_i)\}_i$ and $\{x_S(s_j)\}_j$ have the preceding monotonicity} and
\begin{equation}\label{u:eq:Lambda}
 \Lambda(x):=\sum_{i,j}\revise{w_i z_j}{f_B(i)f_S(j)}\revise{(\phi_i-\psi_j)}{(\psi_B(b_i)-\psi_S(s_j))}x(s_j,b_i)\ge0,
\end{equation}
where $\revise{b_0=b_1}{b_{n+1}=b_n}$, $s_0=s_1$, and $\psi_B$ and $\psi_S$ denote the buyer virtual value and seller virtual cost, respectively:
\begin{equation}\label{u:eq:virtual}
 \revise{\phi_i=b_i-(b_{i-1}-b_i)\frac{q_{i-1}}{w_i}}{\psi_B(b_i)=b_i-(b_{i+1}-b_i)\frac{1-F_B(i)}{f_B(i)}},\qquad
 \revise{\psi_j=s_j+(s_j-s_{j-1})\frac{Q_{j-1}}{z_j}}{\psi_S(s_j)=s_j+(s_j-s_{j-1})\frac{F_S(j-1)}{f_S(j)}}.
\end{equation}
Whenever these conditions hold, pointwise SBB is possible.
\end{lemma}

The cited budget conversion also applies to independent Borel types:
its formula uses only integrable interim payments and their means.
We use the following consequence when passing to general priors.

\begin{lemma}[Common-transfer reconstruction]\label{u:lem:common-transfer}
In bilateral trade with independent Borel types, every BIC, interim-IR mechanism with finite interim payments, integrable transfers, and ex ante WBB can be replaced, without changing its allocation, by a BIC, interim-IR, pointwise-SBB mechanism. The replacement preserves every type's interim allocation and weakly increases the buyer's interim utility. Consequently the feasible allocation sets, and hence the optimal welfare values, are identical under ex ante WBB, pointwise WBB, and pointwise SBB when signed transfers are unrestricted.
\end{lemma}

\subsection{Continuous Buyers and Budget Duality}\label{sec:foundations}
We also work with bounded MHR buyers and finitely supported sellers. This mixed continuous/finite class suffices for the lower bounds; \cref{sec:extension} passes to arbitrary sellers and finite expected welfare.

Let $\revise{X(b)}{x_B(b)}$ and $\revise{Y(s)}{x_S(s)}$ be the buyer's and seller's interim trade probabilities. BIC is equivalent to $\revise{X}{x_B}$ nondecreasing and $\revise{Y}{x_S}$ nonincreasing, together with envelope payments. For a buyer with lower endpoint $b_0$ the maximal interim payment is
\[
 P_B(b)=b\revise{X(b)}{x_B(b)}-\int_{b_0}^b\revise{X(t)}{x_B(t)}\dd t.
\]
For seller types $s_1<\cdots<s_m$ \revise{of masses $z_j>0$}{with probability masses $f_S(j)>0$}, put $\revise{Q_j=\sum_{h\le j}z_h}{F_S(j)=\sum_{h\le j}f_S(h)}$. The minimal seller payment is
\[
 P_S(s_j)=s_j\revise{Y_j}{x_S(s_j)}+\sum_{h=j+1}^m(s_h-s_{h-1})\revise{Y_h}{x_S(s_h)}.
\]
These identities follow by adding adjacent IC inequalities; monotonicity makes the resulting envelope sufficient for every deviation. The lowest buyer and highest seller have zero utility. Adding positive type-independent utilities can only reduce expected surplus. Thus $\Lambda\ge0$ is exactly budget feasibility.

In the exponential coordinate \eqref{eq:mhr}, buyer integration by parts gives the normalized coefficient
\begin{equation}\label{eq:buyer-score}
 u(t)=g(t)-ag'(t).
\end{equation}
At a capped endpoint $g'=0$, so the atom's coefficient is its value. The coefficient is nondecreasing, since $g$ is increasing and $g'$ decreasing. Seller summation gives
\begin{equation}\label{eq:seller-raw}
 \sigma_j=s_j+a(s_j-s_{j-1})\frac{\revise{Q_{j-1}}{F_S(j-1)}}{\revise{z_j}{f_S(j)}},\qquad s_0=s_1.
\end{equation}
Consequently $G(x)+\revise{\lambda}{\alpha}\Lambda(x)=k^{-1}\E[(u-\sigma)x]$.
The coefficient $u$ is integrable: $\E g'(E)=\E g(E)-g(0)$ for bounded $g$, and the same identity holds by monotone approximation when $g'$ is unbounded at zero.

In a matching market the same calculation applies separately to every
agent's interim service. In unnormalized notation, let $\psi_{B_i}$ and $\psi_{S_j}$ denote the buyer-$i$ virtual value and seller-$j$ virtual cost induced by their respective priors; then
$\Lambda(x)=\sum_i\E[\revise{\phi_i(B_i)X_i(B_i)}{\psi_{B_i}(B_i)x_{B_i}(B_i)}]
-\sum_j\E[\revise{\psi_j(S_j)Y_j(S_j)}{\psi_{S_j}(S_j)x_{S_j}(S_j)}]$.
All minimal rents can be imposed simultaneously; matching feasibility
changes the allocation set, not these payment identities.

\begin{lemma}[Common budget duality]\label{lem:duality}
For finite matching markets with downward-closed feasibility, with
either finite priors or bounded MHR buyers and finite sellers,
\begin{equation}\label{eq:strong-dual}
 G_{\SB}=\inf_{\revise{\lambda}{\alpha}\ge0}V_{\revise{\lambda}{\alpha}}.
\end{equation}
\end{lemma}
\begin{proof}[Proof of \cref{lem:duality}]
The finite-prior case is \cite[Lemma~2.1]{LQW26}. For the mixed
continuous/finite case, use compact separation as follows.
Represent a lottery by one probability function for each of the finitely many feasible matchings. These functions belong to a finite product of unit balls of $L^\infty$, are nonnegative, and sum to one pointwise. Pointwise feasibility and monotone interims define a convex weak-* compact set: the simplex constraints are finitely many closed linear inequalities, and monotonicity can be written as nonnegative pairwise interval-average inequalities. A monotone representative can be selected afterwards. Constancy on an atom is another closed linear condition. Expected gains and the envelope surplus are weak-* continuous linear functionals, because their coefficients are integrable. Their image is therefore a compact convex subset $K\subset\R^2$ in coordinates $(\Lambda,G)$, containing $(0,0)$.

Let $P$ maximize $G$ on $K\cap\{\Lambda\ge0\}$. Weak duality gives $P\le\inf_{\revise{\lambda}{\alpha}\ge0}\max_K(G+\revise{\lambda}{\alpha}\Lambda)$. For any $p>P$, strictly separate $K$ from the closed set $\{(l,g):l\ge0,g\ge p\}$. The separating normal has nonnegative coordinates $(c,d)$; moreover $d>0$, since otherwise $(0,0)\in K$ precludes strict separation. Dividing by $d$ gives some $\revise{\lambda}{\alpha}=c/d\ge0$ with $\max_K(G+\revise{\lambda}{\alpha}\Lambda)<p$. Letting $p\downarrow P$ proves the reverse inequality. This argument does not require attainment of a finite dual multiplier.

We may restrict to allocations supported on positive-gain trades when proving the guarantees here. Indeed the ironed fixed-multiplier bilateral optimizers below have this property, and the matching allocations constructed in \cref{sec:matching} do too. The same compact-separation argument applied to the convex closed class of allocations supported on $B_i>S_j$ (zero ties may be removed) turns their all-multiplier lower bounds into budget-feasible guarantees in that class. No assertion that every feasible allocation has this property is needed.
\end{proof}

The dual formulation leaves an optimization over independent value
distributions. We begin with unrestricted priors, where common score
layers reduce that optimization to three parameters.

\section{Sharp Welfare for Unrestricted Priors}\label{sec:unrestricted}
This section identifies the bilateral constant and its \revise{attaining}{extremal}
mechanism. The proof first reduces every affine welfare inequality
to power-law layers, then uses the same layers for the upper bound.
We state the scalar expressions generated by this reduction before the theorem only so that the exact characterization can be stated compactly; their derivation is given in the subsections below. Fixing the budget-dual multiplier $\alpha>0$ produces the normalized parameter $a=\alpha/(1+\alpha)$. After ironing, a common layer-cake decomposition produces a buyer cutoff $r$ and a seller cutoff $p$, and inverting each indicator layer through the reconstruction equation $f+aqf'=h$ yields the power-law pair in \eqref{u:eq:layer-quantiles}. Evaluating one such reconstructed layer gives the endowment $A_a(p)$, the efficient gains $I_a(p,r)$, and the exact dual value $pr/(1-a)$. Thus the formulas below are not postulated independently; they are the scalar coordinates forced by the duality--layer--reconstruction reduction.
For $a,p,r\in(0,1)$ define
\begin{align}
 A_a(p)&=1-\frac{p-ap^{1/a}}{1-a},\label{u:eq:A}\\
 I_a(p,r)&=\int_0^1\min\{1,p(1-t)^{-a}\}
                         \min\{1,rt^{-a}\}\dd t,\label{u:eq:I}\\
 R_{\mathrm U}(a,p,r)&=\frac{A_a(p)+pr/(1-a)}{A_a(p)+I_a(p,r)},\label{u:eq:R}\\
 D(\rho,a,p,r)&=(1-\rho)A_a(p)+pr/(1-a)-\rho I_a(p,r).
 \label{u:eq:D}
\end{align}
Define
\begin{equation}\label{u:eq:rho-star-definition}
 \rho_{\mathrm U}=\inf_{0<a,p,r<1}R_{\mathrm U}(a,p,r).
\end{equation}

The following theorem states that optimization over these reduced-layer parameters is exactly equivalent to the original unrestricted bilateral problem.

\begin{theorem}[Exact unrestricted characterization]\label{u:thm:reduction}
The sharp bilateral welfare ratio equals $\rho_{\mathrm U}$.
The infimum is attained at a unique normalized parameter triple.
The corresponding priors are given in \eqref{u:eq:layer-distributions},
and their optimal mechanism is the boundary lottery of
\cref{u:prop:attaining}.
\end{theorem}
We first prove the variational identity. The global evaluation and
attainment are stated at the end of this section and proved in
\cref{app:unrestricted}.

\subsection{Budget-Dual Ironing}\label{u:sec:ironing}
Fix finite priors and $\revise{\lambda}{\alpha}>0$. We seek monotone virtual coefficients that express the exact dual value and whose cumulative curves can be compared to the original value quantiles. The first property solves the allocation problem at this multiplier; the second will control welfare after reconstruction. Write
\begin{equation}\label{u:eq:a-lambda}
 a=\frac{\revise{\lambda}{\alpha}}{1+\revise{\lambda}{\alpha}}\in(0,1),\qquad k=1-a,
 \qquad 1+\revise{\lambda}{\alpha}=\frac1k.
\end{equation}
The normalized buyer and seller coefficients in the Lagrangian are
\begin{equation}\label{u:eq:normalized}
 \xi_i=k b_i+a\revise{\phi_i}{\psi_B(b_i)}
 =\revise{b_i-a(b_{i-1}-b_i)\frac{q_{i-1}}{w_i}}{b_i-a(b_{i+1}-b_i)\frac{1-F_B(i)}{f_B(i)}},\quad
 \sigma_j=k s_j+a\revise{\psi_j}{\psi_S(s_j)}
 =s_j+a(s_j-s_{j-1})\frac{\revise{Q_{j-1}}{F_S(j-1)}}{\revise{z_j}{f_S(j)}}.
\end{equation}
Equivalently, $\xi_i=k\phi_B^{\alpha}(b_i)$ and $\sigma_j=k\phi_S^{\alpha}(s_j)$, where $\phi_B^{\alpha}(b):=(1+\alpha)b-\alpha[1-F_B(b)]/f_B(b)$ and $\phi_S^{\alpha}(s):=(1+\alpha)s+\alpha F_S(s^-)/f_S(s)$ are the $\alpha$-dependent buyer virtual value and seller virtual cost, and $F_S(s^-)$ denotes the CDF immediately below $s$.
Thus $G(x)+\revise{\lambda}{\alpha}\Lambda(x)=k^{-1}\E[(\xi-\sigma)x]$.

Use descending buyer quantiles $b(q)$ and ascending seller quantiles $s(q)$, for $q\in(0,1]$. For any bounded quantile function $f$ define
\begin{equation}\label{u:eq:H}
 Q_f(q)=\int_0^q f(t)\dd t,\qquad
 H_f(q)=kQ_f(q)+aqf(q).
\end{equation}
Let $K_B$ be the piecewise-linear cumulative curve whose slope on the buyer's $i$th quantile interval is $\xi_i$, and let $K_S$ be the analogous curve with slopes $\sigma_j$. Both start at zero.

Let $\overline K_B$ be the least concave majorant of $K_B$, and $\underline K_S$ the greatest convex minorant of $K_S$. They keep the respective endpoints. Write $u_0$ and $v$ for their slopes, and set $u=(u_0)_+$. The functions $u$ and $v$ are respectively nonincreasing and nondecreasing. All raw seller slopes are nonnegative, and convex-minorant slopes are averages of consecutive raw slopes, so $v\ge0$.

\begin{lemma}[Exact ironed dual value]\label{u:lem:ironing}
For independent \revise{uniform quantiles $q,t$}{quantiles $q,t\sim\mathrm{Unif}(0,1)$},
\begin{equation}\label{u:eq:ironed-V}
 V_{\revise{\lambda}{\alpha}}=\frac1k\int_0^1\int_0^1
       (u(q)-v(t))_+\dd q\dd t.
\end{equation}
Moreover,
\begin{equation}\label{u:eq:primitive-order}
 \int_0^q u(t)\dd t\ge H_b(q),\qquad
 \int_0^q v(t)\dd t\le H_s(q)
\end{equation}
for almost every $q$.
\end{lemma}
\begin{proof}[Proof of \cref{u:lem:ironing}]
After scaling its scores by $k$, the fixed-multiplier formula of
\cite[Lemma~2.2]{LQW26} gives the positive-part expression with slopes $u_0,v$.
Its strict-score allocation is constant on ironing blocks and is
pointwise monotone. Since $v\ge0$, replacing $u_0$ by $(u_0)_+$
does not change that formula. 

To compare cumulative values, telescoping the coefficients gives
$K_B=H_b$ and $K_S=H_s$ at the right endpoints of their quantile intervals.
Within those intervals, direct subtraction gives
\[
 K_B(q)-H_b(q)
 =\revise{a q_{i-1}(b_{i-1}-b_i)
   \left(1-\frac{q-q_{i-1}}{w_i}\right)}{a[1-F_B(i)](b_{i+1}-b_i)
   \left(1-\frac{q-[1-F_B(i)]}{f_B(i)}\right)},
\]
\[
 K_S(q)-H_s(q)
 =\revise{-a Q_{j-1}(s_j-s_{j-1})
   \left(1-\frac{q-Q_{j-1}}{z_j}\right)}{-a F_S(j-1)(s_j-s_{j-1})
   \left(1-\frac{q-F_S(j-1)}{f_S(j)}\right)}.
\]
The jumps are nonnegative, so
\begin{equation}\label{u:eq:discrete-order}
 K_B\ge H_b,\qquad K_S\le H_s
\end{equation}
almost everywhere, irrespective of endpoint conventions. Combining
these inequalities with the majorant and minorant yields
\[
 \int_0^q u\ge\overline K_B\ge K_B\ge H_b,\qquad
 \int_0^q v=\underline K_S\le K_S\le H_s.
\]
\end{proof}

\subsection{Positive Reconstruction}\label{u:sec:reconstruction}
The ironed coefficients describe which trades maximize the dual objective. To compare welfare, we reconstruct value quantiles while controlling their cumulative order. This amounts to inverting the virtual-value equation. For a bounded nonnegative function $h$ on $(0,1]$, define
\begin{equation}\label{u:eq:T}
 (T_a h)(q)=\frac{q^{-1/a}}{a}
             \int_0^q t^{1/a-1}h(t)\dd t.
\end{equation}
This is a positive linear operator. Its weights sum to one, so it preserves bounds. It is the unique bounded, locally absolutely continuous solution of
\begin{equation}\label{u:eq:T-ode}
 f(q)+aq f'(q)=h(q)
\end{equation}
almost everywhere. For monotone $h$, the weighted average in \cref{u:eq:T} lies between its earlier values and its value at $q$. \cref{u:eq:T-ode} then shows that $T_a h$ has the same direction of monotonicity.

Set $\widetilde b=T_a u$ and $\widetilde s=T_a v$. They are valid bounded, nonnegative buyer and seller quantiles. Let $\widetilde A=\E\widetilde S$ and $\widetilde G=\E(\widetilde B-\widetilde S)_+$, with independent quantiles.

\begin{lemma}[Welfare comparison]\label{u:lem:order}
The reconstruction satisfies
\begin{equation}\label{u:eq:welfare-order}
 Q_{\widetilde b}\ge Q_b,\qquad
 Q_{\widetilde s}\le Q_s,\qquad
 \widetilde A\le A,\qquad \widetilde G\ge G.
\end{equation}
Consequently, for every $\beta,\kappa\ge0$,
\begin{equation}\label{u:eq:deficit-order}
 \kappa A+V_{\revise{\lambda}{\alpha}}-\beta G
 \ge \kappa\widetilde A+V_{\revise{\lambda}{\alpha}}-\beta\widetilde G.
\end{equation}
\end{lemma}
\begin{proof}[Proof of \cref{u:lem:order}]
For $f=T_a h$, differentiating $kQ_f+aqf$ and using \cref{u:eq:T-ode} yields
\[
 H_f(q)=\int_0^q h(t)\dd t.
\]
For any bounded $f$, the identity $aqQ_f'(q)+kQ_f(q)=H_f(q)$ has the solution
\begin{equation}\label{u:eq:positive-resolvent}
 Q_f(q)=\frac{q^{-k/a}}a\int_0^q t^{k/a-1}H_f(t)\dd t.
\end{equation}
The integration constant is zero since $Q_f(q)=O(q)$ at zero. The kernel is positive. We apply it to \cref{u:eq:primitive-order} to obtain the first two inequalities of \cref{u:eq:welfare-order}. Evaluating the seller inequality at $q=1$ gives $\widetilde A\le A$.

For a descending buyer quantile and an ascending seller quantile, respectively,
\begin{align}
 \E(B-z)_+&=\max_{0\le q\le1}\{Q_b(q)-zq\},\label{u:eq:stoploss-B}\\
 \E(z-S)_+&=\max_{0\le q\le1}\{zq-Q_s(q)\}.\label{u:eq:stoploss-S}
\end{align}
Indeed, the maximizing prefix consists exactly of the quantiles where the relevant integrand is positive; ties do not affect the integral. Thus the buyer primitive inequality increases the first stop-loss expectation for every $z$, while the seller primitive inequality increases the second. We apply the buyer comparison at each seller value and then the seller comparison at each reconstructed buyer value. Independence permits integrating both comparisons, yielding $\widetilde G\ge G$. \cref{u:eq:deficit-order} follows because $\beta$ and $\kappa$ are nonnegative.
\end{proof}

\subsection{Common Power-Law Layers}\label{u:sec:layers}

We decompose the monotone coefficients $u,v$ into indicator functions at a common threshold and reconstruct each pair with $T_a$. This gives a family of power-law priors. The seller mean and dual value decompose exactly; efficient gains satisfy the inequality needed for \eqref{eq:affine}.
For $t\ge0$, let
$r_t=|\{q:u(q)>t\}|$ and $p_t=|\{q:v(q)\le t\}|$.
Here $|\cdot|$ denotes Lebesgue measure. Up to null sets, monotonicity gives the layer-cake identities
\begin{equation}\label{u:eq:layers}
 u(q)=\int_0^\infty\ind_{\{q<r_t\}}\dd t,\qquad
 v(q)=\int_0^\infty\ind_{\{q>p_t\}}\dd t.
\end{equation}
All these integrals have a finite effective range. Define
$b_{a,r}=T_a\ind_{\{q<r\}}$ and $s_{a,p}=T_a\ind_{\{q>p\}}$.
Explicit integration of \cref{u:eq:T} gives
\begin{equation}\label{u:eq:layer-quantiles}
 b_{a,r}(q)=
 \begin{cases}1,&q\le r,\\(r/q)^{1/a},&q>r,\end{cases}
 \qquad
 s_{a,p}(q)=
 \begin{cases}0,&q\le p,\\1-(p/q)^{1/a},&q>p.\end{cases}
\end{equation}
The corresponding independent priors have, for $0<s,b<1$,
\begin{equation}\label{u:eq:layer-distributions}
 F_S(s)=\min\{1,p(1-s)^{-a}\},\qquad
 \revise{\Prb(B>b)}{1-F_B(b)}=\min\{1,r b^{-a}\}.
\end{equation}
They include an atom $p$ at $S=0$ and an atom $r$ at $B=1$.

\begin{lemma}[Layer quantities]\label{u:lem:layer-quantities}
A layer with parameters $(a,p,r)$ has seller mean $A_a(p)$ and first-best gains $I_a(p,r)$. Moreover,
\begin{align}
 \widetilde A&=\int_0^\infty A_a(p_t)\dd t,\label{u:eq:layer-A}\\
 V_{\revise{\lambda}{\alpha}}&=\int_0^\infty\frac{p_t r_t}{1-a}\dd t,\label{u:eq:layer-V}\\
 \widetilde G&\le\int_0^\infty I_a(p_t,r_t)\dd t.\label{u:eq:layer-G}
\end{align}
\end{lemma}
\begin{proof}[Proof of \cref{u:lem:layer-quantities}]
Integrating the seller quantile in \cref{u:eq:layer-quantiles} yields
\[
 \int_p^1\bigl(1-(p/q)^{1/a}\bigr)\dd q
 =1-\frac{p-a p^{1/a}}{1-a}.
\]
For any independent nonnegative values in $[0,1]$, Fubini's theorem applied to the interval between $S$ and $B$ gives
\begin{equation}\label{u:eq:tail-product}
 \E(B-S)_+=\int_0^1 F_S(t)\revise{\Prb(B>t)}{[1-F_B(t)]}\dd t.
\end{equation}
Changing endpoint conventions affects only a null set of thresholds. This proves the formula for $I_a$.

Linearity and positivity of $T_a$ allow reconstruction inside the integrals in \cref{u:eq:layers}; integrating seller values gives \cref{u:eq:layer-A}. For real nonnegative $u,v$,
\[
 (u-v)_+=\int_0^\infty\ind_{\{v\le t<u\}}\dd t.
\]
Independence of the two quantiles and \cref{u:eq:ironed-V} give \cref{u:eq:layer-V}. Finally, for each quantile pair,
\[
 \left(\int_0^\infty (b_{a,r_t}-s_{a,p_t})\dd t\right)_+
 \le\int_0^\infty(b_{a,r_t}-s_{a,p_t})_+\dd t.
\]
Integrating this inequality gives \cref{u:eq:layer-G}.
\end{proof}

\begin{proof}[Proof of \cref{thm:affine}]
We prove the unrestricted bilateral implication. Suppose $pr/k+\kappa A_a(p)\ge\beta I_a(p,r)$ for all interior $p,r$ at
the fixed $a$. Continuity extends the inequality to $p,r\in[0,1]$.
Reconstruction and the common layers give
\[
 V_{\revise{\lambda}{\alpha}}+\kappa A-\beta G
 \ge\int_0^\infty
 \left[\frac{p_tr_t}{k}+\kappa A_a(p_t)-\beta I_a(p_t,r_t)\right]\dd t
 \ge0.
\]
Conservative rounding in \cref{sec:extension} extends this from finite to
bounded Borel priors. Necessity follows from the exact layer dual value proved
next.
\end{proof}

\subsection{Extremal Priors and the Boundary Lottery}\label{u:subsec:attaining}
The lower bound used the layer expression $pr/(1-a)$ as a contribution to the dual value. We now show that it is also an upper bound on the gains of every feasible mechanism for that layer. A boundary lottery \revise{attains}{achieves} it under an explicit parameter condition. Fix a layer, and abbreviate
\begin{equation}\label{u:eq:xy}
 k=1-a,\qquad x_0=p^{1/a},\qquad y_0=r^{1/a}.
\end{equation}
The continuous seller support is $(0,1-x_0)$; the continuous buyer support is $(y_0,1)$.

\begin{lemma}[Layer budget identity]\label{u:lem:layer-budget}
For any BIC, interim-IR allocation for the layer priors, expected budget feasibility requires
\[
 \Lambda(x)=\E[(\revise{\phi_B(B)}{\psi_B(B)}-\revise{\phi_S(S)}{\psi_S(S)})x(S,B)]\ge0,
\]
where
\begin{equation}\label{u:eq:layer-virtual}
 \revise{\phi_B(b)}{\psi_B(b)}=-\frac{k}{a}b\quad(y_0<b<1),\quad \revise{\phi_B(1)}{\psi_B(1)}=1,
\end{equation}
\begin{equation}\label{u:eq:layer-cost}
 \revise{\phi_S(s)}{\psi_S(s)}=\frac{1-ks}{a}\quad(0<s<1-x_0),\quad \revise{\phi_S(0)}{\psi_S(0)}=0.
\end{equation}
Conversely, monotone interims and $\Lambda(x)\ge0$ suffice for a pointwise-SBB implementation.
\end{lemma}
\begin{proof}[Proof of \cref{u:lem:layer-budget}]
The continuous envelope argument applies on each support interval. Setting boundary rents to zero maximizes expected buyer revenue and minimizes expected seller cost. Fubini gives
\[
 \E P^0_B=\E[B\revise{X(B)}{x_B(B)}]-\int_{y_0}^1\revise{\Prb(B>t)}{[1-F_B(t)]}\revise{X(t)}{x_B(t)}\dd t,
\]
\[
 \E P^0_S=\E[S\revise{Y(S)}{x_S(S)}]+\int_0^{1-x_0}F_S(t)\revise{Y(t)}{x_S(t)}\dd t.
\]
The densities on the continuous parts are
$f_B(b)=ar b^{-a-1}$ and $f_S(s)=ap(1-s)^{-a-1}$.
Thus the continuous buyer coefficient is $b-\revise{\Prb(B>b)}{[1-F_B(b)]}/f_B(b)=-kb/a$, and the seller coefficient is $s+F_S(s)/f_S(s)=(1-ks)/a$. The integrals subtract or add no atomic term; hence the endpoint coefficients are the actual values one and zero. Nonzero boundary utilities can only worsen the budget. Sufficiency follows by the same constant rebate and common-transfer construction as in \cref{u:lem:payments}. For off-support reports extend buyer service constantly above its top support and by zero below its bottom support, and seller service constantly below its bottom support and by zero above its top support; envelope payments implement the extension.
\end{proof}

The layer coefficients now identify which profiles can use budget
surplus without reducing the dual objective. This gives the following
optimal allocation whenever the boundary lottery is feasible.

\begin{proposition}[Optimal boundary lottery]\label{u:prop:attaining}
Every layer satisfies $G_{\SB}\le pr/(1-a)$. If
\begin{equation}\label{u:eq:theta-condition}
 p^{(1-a)/a}+r^{(1-a)/a}\le1,
\end{equation}
then equality holds. In that case let
\begin{equation}\label{u:eq:theta}
 \theta=\frac1{2-p^{(1-a)/a}-r^{(1-a)/a}}.
\end{equation}
Trade surely at $(S,B)=(0,1)$, with probability $\theta$ at every other supported profile on $S=0$ or $B=1$, and never at a profile with $S>0$ and $B<1$. This allocation is BIC, interim-IR, pointwise-SBB implementable and is second-best optimal.
\end{proposition}
\begin{table}[t]
\centering
\begin{tabular}{lcc}
\toprule
Value profile & Lagrangian coefficient & Trade probability\\
\midrule
$S=0$, $B=1$ & $1/(1-a)$ & $1$\\
$S=0$, $y_0\le B<1$ & $0$ & $\theta$\\
$0<S\le1-x_0$, $B=1$ & $0$ & $\theta$\\
$0<S\le1-x_0$, $y_0\le B<1$ & $-1/(1-a)$ & $0$\\
\bottomrule
\end{tabular}
\caption{The optimal boundary lottery and the dual certificate at $\revise{\lambda}{\alpha}=a/(1-a)$. The corner contributes positive budget surplus; randomization on the zero-coefficient boundary profiles uses that surplus without changing the Lagrangian value.}
\label{u:tab:lottery}
\end{table}

\begin{proof}[Proof of \cref{u:prop:attaining}]
\emph{Upper bound.} We choose $\revise{\lambda}{\alpha}=a/k$. The coefficient
$(b-s)+\revise{\lambda}{\alpha}(\revise{\phi_B(b)}{\psi_B(b)}-\revise{\phi_S(s)}{\psi_S(s)})$
equals $1/k$ at $(0,1)$, zero on the remaining parts of the two endpoint lines, and $-1/k$ when both values are in their continuous parts. Every feasible mechanism therefore satisfies
\begin{equation}\label{u:eq:corner-bound}
 G(x)\le G(x)+\revise{\lambda}{\alpha}\Lambda(x)
 =\frac1k\left(pr\,x(0,1)
       -\E[x(S,B)\ind_{\{S>0,\,B<1\}}]\right)
 \le\frac{pr}{k}.
\end{equation}
The four coefficient values and the \revise{attaining}{optimal} allocation are summarized in \cref{u:tab:lottery}.

\emph{Attainment.} Under \cref{u:eq:theta-condition}, $1/2\le\theta\le1$. The proposed allocation is nondecreasing in the buyer's value and nonincreasing in the seller's value, so its interims are monotone. Elementary tail integration gives
$\E[B\ind_{\{B<1\}}]=a(r-y_0)/k$ and $\E[(1-S)\ind_{\{S>0\}}]=a(p-x_0)/k$.
Its budget and gains are consequently
\begin{align*}
 \Lambda&=pr-\theta(2pr-py_0-rx_0)=0,\\
 G&=pr+\theta\frac ak(2pr-py_0-rx_0)=\frac{pr}{k}.
\end{align*}
Here $py_0/(pr)=r^{k/a}$ and $rx_0/(pr)=p^{k/a}$ give the identity for $\theta$. \cref{u:lem:layer-budget} supplies the payments. Equality with \cref{u:eq:corner-bound} proves optimality.
\end{proof}

The corner-only rule \revise{attains}{achieves} the layer Lagrangian $pr/k$ without imposing
budget feasibility. Thus $V_{\revise{\lambda}{\alpha}}=pr/k$ for every layer, even when
\eqref{u:eq:theta-condition} fails. This proves necessity in the unrestricted
part of \cref{thm:affine}.

\begin{proof}[Proof of \cref{u:thm:reduction}]
Let $\eta$ be the infimum in \eqref{u:eq:rho-star-definition}. It is nonnegative
and at most one: send $p\uparrow1$ and then $r\downarrow0$ at fixed $a$,
using $A_a(1)=0$ and $I_a(1,r)=(r-ar^{1/a})/k$.
By definition, every layer satisfies
\[
 pr/k+(1-\eta)A_a(p)\ge\eta I_a(p,r).
\]
The unrestricted part of \cref{thm:affine}, budget duality, and
conservative rounding give
$A+G_{\SB}\ge\eta(A+G)$ for every bounded bilateral instance.
\Cref{sec:extension} removes the boundedness assumption.
Conversely, the exact layer dual value and weak duality give
\[
 \frac{A+G_{\SB}}{A+G}
 \le \frac{A_a(p)+pr/k}{A_a(p)+I_a(p,r)}
\]
on every layer. Taking the infimum proves equality.
Attainment and uniqueness follow from \cref{u:thm:scalar} and the
lottery-condition check immediately after its statement.
\end{proof}

\subsection{Global Evaluation and Attainment}
It remains to locate the scalar minimum and verify the lottery
condition there. The global certificate yields both conclusions.

\begin{theorem}[Scalar minimization]\label{u:thm:scalar}
The function $R_{\mathrm U}$ has a unique global minimizer $(a_*,p_*,r_*)$ in $(0,1)^3$, and
$0.8882516903<\rhostar<0.8882516904$.
Moreover, $(\rhostar,a_*,p_*,r_*)$ is the unique solution of
\begin{equation}\label{u:eq:stationary}
 D=\partial_aD=\partial_pD=\partial_rD=0
\end{equation}
in the box
\begin{equation}\label{u:eq:smallbox}
 (0.8882516,0.8882517)\times(0.323,0.324)
 \times(0.654,0.655)\times(0.325,0.326).
\end{equation}
Numerically, $(a_*,p_*,r_*)\approx(0.3231442122,0.6544790312,0.3251658297)$.
\end{theorem}

The full finite certificate in \cref{app:unrestricted} proves the displayed enclosure and a stronger enclosure recorded in \cref{u:subsec:point}.
At the minimizing triple, \eqref{u:eq:smallbox} gives $a<.324<1/3$,
$p<.655$, and $r<.326$. Thus
$p^{k/a}+r^{k/a}<.655^2+.326^2<1$, so the boundary lottery \revise{attains}{achieves}
the exact ratio $\rho_{\mathrm U}$.

This completes the unrestricted bilateral characterization and its
\revise{attaining}{extremal} example. We next impose the MHR restriction on buyers.

\section{Sharp Welfare for MHR Buyers}\label{sec:mhr-start}
The power-law buyer layers above are not generally MHR. We must preserve
concavity in the exponential coordinate while obtaining an exact dual
decomposition, an exact seller mean, and an upper decomposition of efficient
gains. The seller reconstruction from the preceding section remains
available because it leaves the buyer unchanged.

\begin{lemma}[One-sided positive reconstruction]\label{lem:seller-reconstruction}
For a bounded MHR buyer and any finite seller, the seller can be replaced
by a nonnegative seller with a nondecreasing normalized score, the same $V_{\revise{\lambda}{\alpha}}$, no larger mean,
and no smaller efficient gains.
\end{lemma}
\begin{proof}[Proof of \cref{lem:seller-reconstruction}]
Keep the monotone buyer score $u=g-ag'$ unchanged. The seller half of
\cref{u:lem:ironing} applies verbatim: if $v$ is the
nondecreasing slope of the convex minorant of the seller cumulative score,
the strict threshold rule is constant on its ironing intervals and \revise{attains}{achieves}
\begin{equation}\label{eq:score-V}
 V_{\revise{\lambda}{\alpha}}=\frac1k\E(u-v)_+.
\end{equation}
Set $\widetilde s=T_av$, using the operator in \eqref{u:eq:T}.
Its normalized seller score is $v$, so the optimized dual is unchanged.
The seller primitive order and stop-loss identity in \cref{u:lem:order}
give $\E\widetilde S\le\E S$ and
$\E(B-\widetilde S)_+\ge\E(B-S)_+$.
\end{proof}

\subsection{MHR-Preserving Decomposition}\label{sec:localization}
With the seller reconstructed, common transformations of the two
scores preserve the dual value and seller mean. This is the step
that reduces the buyer to shifted capped exponentials.

\begin{proof}[Proof of \cref{thm:affine}]
We prove the MHR bilateral reduction. Necessity is immediate. A buyer identically zero has $G=V_{\revise{\lambda}{\alpha}}=0$, so its inequality is automatic. For every other buyer, first reconstruct the seller as in \cref{lem:seller-reconstruction}. This preserves $V_{\revise{\lambda}{\alpha}}$, decreases $A$, and increases $G$, so it can only decrease the deficit $V_{\revise{\lambda}{\alpha}}+\kappa A-\beta G$.

Suppose the buyer score crosses zero at coordinate $r$. Concavity and nonnegativity give $r\le a$: at a crossing, $g(r)\le ag'_-(r)$, while $g(r)\ge rg'_-(r)$. Take the supporting slope $m=g(r)/a$, between the one-sided derivatives, and replace $g$ below $r$ by the supporting line
\[
 \widetilde g(t)=g(r)+m(t-r),\qquad t<r.
\]
This raises buyer values, preserves concavity and all positive scores, and makes the lower scores $m(t-r)\le0$. The seller score is nonnegative, so $V_{\revise{\lambda}{\alpha}}$ is unchanged, whereas $G$ increases. A zero crossing at an endpoint is understood by a limit.

The curvature measure of $\widetilde g$ gives
\begin{equation}\label{eq:g-decompose}
 \widetilde g(t)=\int_{H\ge r}[a-r+\min(t,H)]\,\mu(\dd H).
\end{equation}
To see this directly, set $\mu((t,\infty])=\widetilde g'_+(t)$; an atom at infinity represents a remaining linear slope. Integrating these slopes and matching the value $(a-r)m$ at zero proves the identity. Each component has score
\[
 u_H(t)=\begin{cases}t-r,&t<H,\\H+a-r,&t\ge H.\end{cases}
\]
The score $u$ is strictly increasing until its final constant part, and every component score is constant on that final part. There are consequently nondecreasing maps $\psi_H$ satisfying $\psi_H(u(t))=u_H(t)$. On a gap in the range of $u$, interpolate each map between its endpoint values. Their weighted sum is the identity at the endpoints and therefore throughout the gap. Above the final score extend all maps proportionally. In particular,
\begin{equation}\label{eq:common-maps}
 \int\psi_H(z)\mu(\dd H)=z,\qquad \psi_H(z)\ge0\quad(z\ge0).
\end{equation}
Below the reserve all component scores are nonpositive, and the common linear definition suffices. A jump at the reserve starts at score zero after the supporting-line replacement, so the interpolation is also nonnegative there.

Apply these same maps to the seller: $v_H=\psi_H(v)$ and $s_H=T_av_H$. Monotonicity implies, pointwise,
\[
 (u-v)_+=\int(\psi_H(u)-\psi_H(v))_+\mu(\dd H).
\]
Linearity of $T_a$ and subadditivity of the positive part now give
\begin{equation}\label{eq:exact-decompose}
 A=\int A_H\dd\mu,\qquad V_{\revise{\lambda}{\alpha}}=\int V_{\revise{\lambda}{\alpha},H}\dd\mu,
 \qquad G\le\int G_H\dd\mu.
\end{equation}
Every component buyer has $\ell=a-r$ and $U=H+a-r\ge a$, as required.

If the initial score is positive, put $d=g(0)-ag'_+(0)>0$ and write
\[
 g(t)=d+\int[a+\min(t,H)]\mu(\dd H).
\]
There is one extra deterministic buyer component. On $0\le z\le d$ assign its score map the value $z$ and assign all the other maps zero; on the actual buyer-score range assign it the constant $d$. Interpolate remaining gaps as before. The same identities hold. A deterministic component is a positive scaling of the limit $\ell=U=a$. Thus this argument remains valid for arbitrary nonnegative $\beta,\kappa$, not merely for $\beta\le1$.

All integrands used for gains, means and positive score differences are nonnegative, so Tonelli justifies the integrals. Equivalently one can first use a finite piecewise-linear concave $g$, then approximate its slope measure monotonically. Formula \eqref{eq:exact-decompose} proves the desired deficit inequality by integration. Seller approximation in \cref{sec:extension} removes the finite-support restriction. The case $\revise{\lambda}{\alpha}=0$ is handled separately by $V_0=G$.
\end{proof}

\subsection{Typewise Welfare Certificates}\label{sec:typewise}
The decomposition leaves one \revise{canonical}{normalized capped-exponential} buyer and an arbitrary seller.
We now seek an allocation whose welfare inequality holds at every
seller type; integrating will then cover every seller prior.

For an MHR buyer let $C_B(\eta)$ be the integral of its normalized score over the top $\eta$ buyer probability mass. For a nonincreasing service function $p(s)\in[0,1]$ of bounded support, serve that top buyer mass at seller report $s$. This rule is pointwise monotone. Normalize the buyer envelope at its lower endpoint and the seller envelope above the support of $p$. Direct integration of those payments gives the Lagrangian value
\begin{equation}\label{eq:typewise}
 \frac1k\E_S\left[C_B(p(S))-Sp(S)-a\int_S^\infty p(t)\dd t\right].
\end{equation}
This is \revise{an attainable value}{a feasible Lagrangian value} and therefore a lower bound on $V_{\revise{\lambda}{\alpha}}$. If the seller support ends before the support of $p$, these payments may leave a positive utility to its highest type. We do \emph{not} identify \eqref{eq:typewise} with the maximal payment surplus in that case. We use the \revise{attainable}{feasible} expression for lower bounds and check its optimality separately for the hard seller.

Thus the pointwise inequality
\begin{equation}\label{eq:pointwise}
 C_B(p(s))-sp(s)-a\int_s^\infty p(t)\dd t
 \ge k[\beta H_B(s)-\kappa s],\qquad H_B(s)=\E(B-s)_+,
\end{equation}
is a lower certificate against every seller prior, including atoms.
For \eqref{eq:canonical}, put $q=e^{\ell-U}$. Then
\begin{equation}\label{eq:CB}
 C_B(\eta)=\begin{cases}
 U\eta,&0\le\eta\le q,\\
 \eta(\ell-\log\eta+k)-kq,&q\le\eta\le1.
 \end{cases}
\end{equation}

\subsection{The Common Residual Equation}\label{sec:residual}
The \revise{canonical buyer}{capped-exponential buyer in \eqref{eq:canonical}} has two continuous parameters, $\ell$ and $U$,
in addition to the multiplier $a$. We will construct a function
$\revise{r(a,\ell,U)}{R_{\mathrm M}(a,\ell,U)}$ satisfying
\begin{equation}\label{eq:w-main}
 \rho_{\mathrm M}
 =\inf_{0<a<1,\ 0<\ell<a<U}\revise{r(a,\ell,U)}{R_{\mathrm M}(a,\ell,U)}
 =\min_{0<a<1,\ 0<\ell<a<U}\revise{r(a,\ell,U)}{R_{\mathrm M}(a,\ell,U)}.
\end{equation}
The equality with the infimum follows from the next proposition and
localization; attainment and all boundary cases are proved in
\cref{sec:global}. The function is defined by a scalar residual
$\mathcal R$ rather than by numerical optimization.

Fix $0<a<1$, $0\le\ell<a<U$, and a trial $0<\rho<1$. Set
\[
 \kappa=1-\rho,\quad q=e^{\ell-U},\quad h=U-a,\quad
 \revise{Q(s)=}{\overline F_B(s):=}\min\{1,e^{\ell-s}\}\ (s<U),\quad b(s)=\kappa+\rho\revise{Q(s)}{\overline F_B(s)},
\]
Here $\overline F_B(s)=\Prb(B>s)=1-F_B(s)$ is the buyer survival function on the continuous range. The buyer stop-loss is
\begin{equation}\label{eq:stoploss}
 H(s)=\begin{cases}\ell-s+1-q,&0\le s\le\ell,\\
 e^{\ell-s}-q,&\ell<s<U,\\0,&s\ge U.
 \end{cases}
\end{equation}
Define the head cost in seller-score coordinates by writing $s_y:=\mathrm ds/\mathrm dy$ and setting
\begin{equation}\label{eq:head}
 s_y=\frac{y-s}{D(y,s)},\qquad s(0)=0,\qquad
 D(y,s)=k(e^{y+a-\ell}b(s)-1),\quad0\le y\le h.
\end{equation}
On $0\le s\le y$, the denominator is positive: below $\ell$ it is at least $k(e^{a-\ell}-1)$, and above $\ell$ at least $k(e^a-1)$. The vector field is locally Lipschitz and continuous at the buyer floor. The invariant region gives $0<s(y)<y$ and $s_y>0$ for $y>0$. Existence, uniqueness and continuous parameter dependence follow on every finite interval. Set $z=s(h)$.

Let $L$ be the unique solution of $\kappa L=\rho H(L)$. Its left side increases and its right side decreases, so $0<L<U$. Define
\begin{equation}\label{eq:R}
 \mathcal R(a,\ell,U;\rho)=q(U-z)^k-k\int_z^Lb(s)(U-s)^{-a}\dd s,
\end{equation}
with the usual oriented-integral convention when $z>L$.

The residual measures the mismatch between the continuous buyer-tail
rule and its top-atom continuation. At a zero they join, and a seller
distribution makes the resulting lower certificate tight.

\begin{proposition}[Exact seller optimization]\label{prop:root}
For each $(a,\ell,U)$ above, $\mathcal R$ has a unique zero $\revise{r(a,\ell,U)}{R_{\mathrm M}(a,\ell,U)}$ in $(0,1)$. It is positive below that zero and negative above it. Moreover
\begin{equation}\label{eq:seller-inf}
 \revise{r(a,\ell,U)}{R_{\mathrm M}(a,\ell,U)}=\inf_{S\ {\rm bounded}}\frac{A+V_{\revise{\lambda}{\alpha}}}{A+G}.
\end{equation}
The infimum is attained by the explicit seller in \eqref{eq:hard-seller} below.
\end{proposition}
\begin{proof}[Proof of \cref{prop:root}]
For existence one can divide $\mathcal R$ by $(U-z)^k$ and substitute $t=(U-s)^k/(U-z)^k$, obtaining
\[
 q-\int_{((U-L)/(U-z))^k}^1
 b\bigl(U-(U-z)t^{1/k}\bigr)\dd t.
\]
At $\rho=0$, $L=0<z$ and this is positive. At $\rho=1$, $L=U$ and $b=\revise{Q}{\overline F_B}$; on the integration interval $\revise{Q}{\overline F_B}>q$ almost everywhere, so the expression is negative. Continuity, including these integrable endpoint limits, gives a zero. Every zero has $z<L$.

At a zero, invert the strictly increasing head and put $w(s)=y(s)-s$. Define
\begin{equation}\label{eq:service}
 p(s)=\begin{cases}
 e^{\ell-a-y(s)},&0\le s<z,\\
 \displaystyle\frac{k}{(U-s)^k}\int_s^Lb(t)(U-t)^{-a}\dd t,&z\le s\le L,\\
 0,&s>L.
 \end{cases}
\end{equation}
The zero condition joins the pieces at $p(z)=q$. On the head $p$ decreases. On the tail,
\[
 p(s)\le b(s)\left[1-\left(\frac{U-L}{U-s}\right)^k\right]<b(s),
 \qquad p'(s)=\frac{k(p(s)-b(s))}{U-s}<0.
\]
The head has $q\le p(s)\le e^{\ell-a}<1$, and the tail serves only the top atom. Differentiating the deficit in \eqref{eq:pointwise} gives
\begin{equation}\label{eq:euler}
 [C'_B(p(s))-s]p'(s)=k(p(s)-b(s)).
\end{equation}
On the head this is equivalent to \eqref{eq:head}; on the tail it follows from \eqref{eq:service}. At $L$ equality in \eqref{eq:pointwise} follows from $\kappa L=\rho H(L)$. Thus equality holds on $[0,L]$, and above $L$ zero service suffices because $\rho H(s)-\kappa s\le0$. This proves a lower bound of $\rho$ against every seller.

Write $d=U-L$ and $W=U-z$. The matching seller is
\begin{equation}\label{eq:hard-seller}
 \revise{F_\rho(s)}{F_{S,\rho}(s)}=\begin{cases}
 \displaystyle(d/W)^a\exp\left(-\int_s^z\frac a{w(t)}\dd t\right),&0\le s<z,\\
 (d/(U-s))^a,&z\le s\le L,\\
 1,&s>L.
 \end{cases}
\end{equation}
It is a \revise{valid CDF}{valid seller CDF} with positive zero-cost atom. For brevity in the calculations below, write $F(s):=F_{S,\rho}(s)$. To verify finiteness at zero without a singular integral, put $J(y)=\int_0^y a/D(v,s(v))\dd v$; then
\[
 \revise{F_\rho(s(y))}{F_{S,\rho}(s(y))}=(d/W)^a e^{J(y)-J(h)}.
\]
Its continuous normalized score is $s+aF/F'=y(s)$ on the head, and $U$ on the tail. The atom at zero has score zero. These scores increase, with an upward jump from $U-a$ to $U$ at the join. The rule \eqref{eq:service} maximizes the Lagrangian: it accepts positive score differences on the head and randomizes only at the zero-score top-atom/tail ties. The seller's highest type has zero service and utility. Hence \eqref{eq:typewise} is now the exact optimized dual value, and the typewise equality gives
\[
 V_{\revise{\lambda}{\alpha}}+(1-\rho)A=\rho G.
\]
This identifies every zero with the same seller infimum, proving uniqueness. The endpoint signs and continuity give the sign characterization. The lower bound at every multiplier and the weak-duality upper witness, combined with localization, budget duality, and approximation, prove \eqref{eq:w-main} with an infimum in place of a minimum. Attainment and the interior location follow in \cref{sec:global}.
\end{proof}

\subsection{The Positive Buyer Floor}
The buyer's lower endpoint is essential for welfare. For a fixed seller
with nondecreasing nonnegative normalized scores and a fixed cap $U$,
when $0\le\ell\le a$,
\[
 V_{\revise{\lambda}{\alpha}}(\ell)=e^\ell V_{\revise{\lambda}{\alpha}}(0),\qquad
 G(\ell)=\int_0^\ell F(s)\dd s+
         e^\ell\int_\ell^UF(s)e^{-s}\dd s.
\]
The first identity follows from the score survival above zero; the
second is the tail-product formula. At an interior minimizing zero,
the derivative of $A+V_{\revise{\lambda}{\alpha}}-\rho(A+G)$ in $\ell$ vanishes, since
the same seller is an admissible trial instance after a small change
in $\ell$. Substituting the zero-deficit identity gives
\begin{equation}\label{eq:moment-balance}
 (1-\rho)A=\rho\E(\ell-S)_+.
\end{equation}
At $\ell=0$, that derivative is instead $-(1-\rho)A<0$ when
$A>0$ and $\rho<1$, so an unshifted capped exponential is strictly
improvable by raising its floor. This rules out the natural attempt to
search only over unshifted exponentials.

\paragraph{Global evaluation.}
The next theorem locates the minimum in \eqref{eq:w-main} and completes
the sharp MHR characterization.

\begin{theorem}[MHR minimization]\label{thm:certification}
The function $\revise{r(a,\ell,U)}{R_{\mathrm M}(a,\ell,U)}$ has a unique global minimizer in
$0<a<1$, $0<\ell<a<U$, and its minimum $\rho_{\mathrm M}$ satisfies
\eqref{eq:sharp-bracket}. More precisely, the minimizing parameters
and $\rho_{\mathrm M}$ form the unique solution in
$\mathcal N\times(.911389368124,.911389368127)$ of
\begin{equation}\label{eq:stationary-system}
 \mathcal R(\theta;\rho)=0,\qquad
 \nabla_\theta\mathcal R(\theta;\rho)=0,
 \qquad \theta=(a,\ell,U),
\end{equation}
where $\mathcal N$ is the box in \eqref{eq:N}. The corresponding buyer
and seller in \eqref{eq:canonical} and \eqref{eq:hard-seller} attain
the second-best welfare ratio.
\end{theorem}

For orientation, the minimizing coordinates are
\[
 (a_*,\ell_*,U_*)\simeq
 (.2852306465,.0494569729,1.1746687514).
\]
The exact parameters, rather than these rounded coordinates, define the
\revise{attaining}{extremal} buyer in \eqref{eq:canonical} and seller in
\eqref{eq:hard-seller}. The proof combines a finite-domain reduction,
global interval exclusions, and local strict convexity. The derivative
identities and complete certification argument are in
\cref{app:mhr}.

We have now obtained sharp bilateral affine bounds in both prior
classes. Their extension to markets requires preserving the seller
endowment term when local trade opportunities are combined.

\section{Welfare Transfer via Bilateral Lifting}\label{sec:matching}
The bilateral analysis has established affine inequalities of the form
$V_{\revise{\lambda}{\alpha}}+\kappa A\ge\beta G$. We now apply the lifting framework of
Liu, Qin, and Wang~\cite{LQW26}. Their edge-by-edge composition theorem
allows general typewise targets, but its GFT corollary does not account
for $A$. The additional task here is to charge each seller's endowment
at most once. For unrestricted priors this gives the full downward-closed
extension. For MHR buyers we use a separate argument on ordinary
compatibility graphs to preserve the distributional restriction.

Throughout, the \revise{canonical first-best matching}{selected first-best matching $M^*(\theta)$ at type profile $\theta$} maximizes gains, then
minimizes cardinality, then uses a fixed total order of matching
identities. It therefore uses only positive-gain edges.

\subsection{Unrestricted Priors and Endowment Charges}
We use three results from~\cite{LQW26}: Theorem~3.3 converts bilateral
value bounds, \revise{uniform over}{valid for} all buyer probability vectors, into monotone
local allocation rules that increase with the seller cap; Theorem~3.2
composes these rules at the first-best caps; and Proposition~3.4
regularizes sellers before the caps are chosen. Their incentive and
composition arguments need not be repeated. We verify the welfare
targets and the resulting endowment charge.

\begin{proof}[Proof of \cref{thm:affine}]
We prove the unrestricted matching implication. Fix $\revise{\lambda}{\alpha}>0$ and
suppose \eqref{eq:affine} holds for every finite bilateral instance. Start with finite priors and a downward-closed
family $\mathcal F$ of feasible matchings.

\emph{Seller regularization.}
Proposition~3.4 of~\cite{LQW26} replaces each full seller prior by a
$\revise{\lambda}{\alpha}$-weakly regular prior: its \revise{combined cost $s+\lambda\psi(s)$}{$\alpha$-virtual cost $\phi_S^{\alpha}(s)=s+\alpha\psi_S(s)$}
is nondecreasing. Denote the original and transformed instances by $\mathcal I$ and $\widehat{\mathcal I}$, respectively. The transformed market satisfies
\begin{equation}\label{eq:market-regularization}
 V_{\revise{\lambda}{\alpha}}(\mathcal I)=V_{\revise{\lambda}{\alpha}}(\widehat{\mathcal I}),\qquad
 \widehat G\ge G,\qquad \widehat A\le A.
\end{equation}
The first two comparisons are the cited proposition. The third follows
by evaluating the cumulative seller-quantile inequality in its
Appendix~A.3 at quantile one. The reconstruction keeps the lowest cost
and has increasing support, so costs remain nonnegative. It suffices
to prove the affine bound for these weakly regular sellers.

\emph{Local targets.}
Fix such a seller with costs $s_1<\cdots<s_m$, \revise{masses $f_k$}{probability masses $f_S(k)$}, and full
mean $\mu=\sum_k \revise{f_k}{f_S(k)}s_k$. For the cap consisting of its first $r$ types,
put $\revise{F_r=\sum_{k\le r}f_k}{F_S(r)=\sum_{k\le r}f_S(k)}$ and define the unnormalized targets
\begin{equation}\label{eq:full-prior-target}
 g_r(b)=\sum_{k\le r}\revise{f_k}{f_S(k)}(b-s_k)_+,\qquad
 C_r(b)=[\beta g_r(b)-\kappa\mu]_+,\qquad C_0=0.
\end{equation}
They are nonnegative and nondecreasing with the cap. Using the full
mean $\mu$ keeps the charge fixed as the cap expands.

To check the premise of \cite[Theorem~3.3]{LQW26}, take any buyer
probability vector on any finite support and a nonempty cap. Let
$U_r=\{b:\beta g_r(b)>\kappa\mu\}$, an upper set. If it has zero
probability, the required bound is immediate from no trade. Otherwise
apply the assumed bilateral inequality conditional on $B\in U_r$ and
$S\le s_r$. The pointwise-monotone optimizer of
\cite[Lemma~2.2]{LQW26}, extended by zero outside this rectangle, is
feasible for the buyer's full support and the capped seller.
Upper-tail buyer conditioning and lower-prefix seller conditioning
leave the retained raw virtual coefficients unchanged: both the
relevant tail or prefix and its point mass have the same normalizing
factor. Consequently
\begin{align}
 \revise{F_r}{F_S(r)}V_{\revise{\lambda}{\alpha}}(B,S\mid S\le s_r)
 &\ge \beta\E_B[\ind_{U_r}g_r(B)]
       -\kappa\Prb(B\in U_r)\sum_{k\le r}\revise{f_k}{f_S(k)}s_k \notag\\
 &\ge \E_B C_r(B).\label{eq:target-verification}
\end{align}
The last inequality uses $\sum_{k\le r}\revise{f_k}{f_S(k)}s_k\le\mu$.
This upper-tail argument justifies the positive part in
\eqref{eq:full-prior-target}; taking the positive part of an expected
inequality alone would not suffice. Theorem~3.3 now supplies the local
rules for all caps, with the typewise targets $C_r$.

\emph{Composition and charging.}
For seller $j$, write $\mu_j:=\E S_j$ and let $C_r^{(j)}$ denote the target in \eqref{eq:full-prior-target} constructed from seller $j$'s prior and cap $r$.
Let $r_{ij}(\theta_{-j})$ be the first-best seller cap for edge $(i,j)$,
with every report except $s_j$ fixed, where $\theta_{-j}$ denotes the type profile of all agents other than seller $j$. Theorem~3.2 of~\cite{LQW26} gives
\begin{equation}\label{eq:welfare-cap-composition}
 V_{\revise{\lambda}{\alpha}}\ge\sum_{(i,j)\in\mathcal E}
 \E_{\theta_{-j}}C^{(j)}_{r_{ij}(\theta_{-j})}(B_i).
\end{equation}
By the first-best threshold property in \cite[Lemma~3.1]{LQW26},
a seller's partner is fixed throughout her winning reports. Thus, for
each fixed $\theta_{-j}$, at most one incident edge has a nonempty cap:
\begin{equation}\label{eq:one-cap-charge}
 \sum_i\ind_{\{r_{ij}(\theta_{-j})\ne0\}}\le1.
\end{equation}
For every cap,
$C_r(b)\ge\beta g_r(b)-\kappa\mu_j\ind_{\{r\ne0\}}$.
The $g_r$ terms in \eqref{eq:welfare-cap-composition} sum in expectation
to $G$, since the caps describe exactly the positive-gain edges of the
\revise{canonical}{tie-broken} first best. By \eqref{eq:one-cap-charge}, the mean charges
sum to at most $\kappa\sum_j\mu_j=\kappa A$. This proves
$V_{\revise{\lambda}{\alpha}}\ge\beta G-\kappa A$ in the regularized market.
Returning through \eqref{eq:market-regularization} proves it in the
original market as well.

Conservative rounding in \cref{sec:extension} extends this implication
to bounded Borel priors. That argument uses only downward closure and
a finite maximum matching size. The reverse implication follows by
taking a single-edge market.
\end{proof}

\subsection{MHR Buyers and Rectangle Packing}
The automatic local-rule theorem just used quantifies over every buyer
probability vector. An MHR-only bilateral bound does not meet that
premise, because arbitrary reweighting need not preserve MHR.
We therefore retain a direct construction for ordinary matching
markets. Its conditioning uses only buyer upper tails, which preserve
MHR, and its geometric property controls the endowment charges.

\begin{lemma}[Rectangle packing]\label{lem:rectangles}
Suppose every matching of $\mathcal E$ is feasible. Fix an edge
$e=(i,j)$ and all other types $\omega$. There is a possibly empty
rectangle $Q_e(\omega)$, an upper set in $b_i$ and a lower set in $s_j$,
such that
\begin{equation}\label{eq:rect-event}
 \{e\in M^*(b_i,s_j,\omega)\}=Q_e(\omega)\cap\{b_i>s_j\}.
\end{equation}
At any full type profile, all edges whose rectangles contain that
profile form a matching, including rectangle edges with nonpositive
gains. For an edge $e=(i,j)$, write $j(e):=j$ for its seller endpoint. Consequently
\begin{equation}\label{eq:charge}
 \sum_e S_{j(e)}\ind_{Q_e}\le\sum_jS_j.
\end{equation}
\end{lemma}
\begin{proof}[Proof of \cref{lem:rectangles}]
Classify matchings into those containing $e$, touching neither endpoint, touching only $i$, touching only $j$, and touching both endpoints on different edges. Their best weights have forms
\[
 b-s+A_e,\quad A_{00},\quad b+A_{10},\quad -s+A_{01},\quad b-s+A_{11},
\]
respectively. Here $A_e,A_{00},A_{10},A_{01},A_{11}$ are the type-independent residual optimal matching weights for these five categories. An empty category has coefficient $-\infty$. Ordinary matching gives $A_e=A_{00}$ by adding or removing $e$ from a matching avoiding its endpoints. Comparisons with the middle two categories impose a lower threshold on $b$ and an upper threshold on $s$; the last comparison is type-independent. If this last comparison loses, take an empty rectangle. Fixed tie rules specify open or closed endpoints. On the resulting rectangle, the best matching among those touching at least one endpoint contains $e$. Its comparison with the neither-endpoint category is precisely $b>s$: at equality, removing $e$ preserves gains and decreases cardinality. This proves \eqref{eq:rect-event}.

For packing, suppose two distinct incident edges $e,f$ have rectangle membership. Let $N_e,N_f$ be the \revise{canonically best}{tie-broken best} matchings among those touching an endpoint of the respective edge. Each contains its defining edge. Because of the shared endpoint, each matching is a competitor in the other's optimization. The common strict ranking forces $N_e=N_f$, impossible for a matching containing both incident edges. Thus rectangle edges are disjoint, proving \eqref{eq:charge} for nonnegative seller values.
\end{proof}

The lemma lets us apply the bilateral welfare inequality on independent
conditional priors. The remaining check is that these local optimizers
assemble into a monotone allocation with the same payment coefficients.

\begin{proof}[Proof of \cref{thm:affine}]
We prove the MHR matching implication. Start with bounded MHR buyers
and finite sellers, and fix $\revise{\lambda}{\alpha}>0$.
For each edge and residual profile, condition its endpoints on
$Q_e(\omega)$, ignoring zero-probability rectangles. Buyer upper-tail
conditioning preserves MHR: it shifts the exponential coordinate and
leaves the hazard unchanged on the retained support. Seller
lower-prefix conditioning is unrestricted.

Choose the strict-score monotone bilateral optimizer and extend it by
zero outside the rectangle. The original and conditional normalized
coefficients agree at retained types. For buyers this follows from
hazard invariance, including the value coefficient of an upper atom;
for sellers the ratio $\revise{Q_{j-1}/z_j}{F_S(j-1)/f_S(j)}$ in \eqref{eq:seller-raw} is unchanged.
The buyer score is at most its value, and the ironed seller score is at
least its cost, by the seller endpoint bound in
\cite[Appendix~A.2]{LQW26}. Hence every selected trade has positive
gains. By \eqref{eq:rect-event} the assembled allocation is supported
on the \revise{canonical}{tie-broken} first-best matching, and is feasible.

The fixed-partner property of \cite[Lemma~3.1]{LQW26} also gives
pointwise monotonicity. Within an agent's first-best trading region,
its partner, residual profile, conditioning rectangle, and local rule
are fixed, and the local service probability is monotone. Outside that
region the agent is unserved. This applies to both buyers and sellers.
The strict-score rule, finite seller ironing, and the finite set of
\revise{canonical matching}{fixed tie-breaking} comparisons give measurable choices in the
residual reports.

Apply the bilateral inequality to each conditional pair, multiply by
its rectangle probability, and average over residual types. The
preserved payment coefficients yield
\begin{align*}
 G(x)+\revise{\lambda}{\alpha}\Lambda(x)
 &\ge\beta\sum_e\E[(B_i-S_j)_+\ind_{Q_e}]
       -\kappa\sum_e\E[S_j\ind_{Q_e}]\\
 &\ge\beta G-\kappa A.
\end{align*}
The first sum equals $G$ by \eqref{eq:rect-event}; the second is bounded
by $A$ by \eqref{eq:charge}. Taking the supremum proves the affine
bound. Rounding sellers as in \cref{sec:extension} extends it to all
bounded seller priors without discretizing the MHR buyers. A single
edge gives the reverse implication.
\end{proof}

Here $A_e=A_{00}$ uses the ability to add $e$ to any matching avoiding
its endpoints. Additional downward-closed constraints can prevent
that operation. Thus the MHR argument establishes the ordinary-matching
claim; it does not establish the stronger feasibility extension proved
above for unrestricted priors.

\subsection{Sharp Welfare Ratios}
The transfers preserve the affine coefficients at every multiplier.
Budget duality therefore gives the matching lower bounds, while the
bilateral \revise{attaining}{extremal} instances already supply their upper bounds.

\begin{proof}[Proof of \cref{thm:main}]
\Cref{u:thm:reduction,u:thm:scalar} give the unrestricted bilateral
constant. \Cref{prop:root,thm:certification}, MHR localization, and
budget duality give the MHR bilateral constant. Apply the preceding
transfers with $\beta=\rho$ and $\kappa=1-\rho$ at every $\revise{\lambda}{\alpha}>0$;
at $\revise{\lambda}{\alpha}=0$, $V_0=G$. Taking the infimum over multipliers as in
\cref{lem:duality} gives the unrestricted guarantee for downward-closed
matching markets and the MHR guarantee for ordinary matching markets.
In the finite unrestricted case, \cite[Lemma~2.2]{LQW26} \revise{attains}{achieves} the
dual value using only positive-gain edges, so the restricted separation
argument in \cref{lem:duality} applies as well. The MHR construction
already has this property.

Single-edge markets give sharpness in both classes.
\Cref{sec:extension} extends the guarantees to finite expected welfare,
removes the buyer's upper atom without changing the infimum, and gives
pointwise SBB. At the unrestricted minimizing triple the boundary
lottery \revise{attains}{achieves} the bound. At the MHR minimizing triple, the hard
seller's dual upper bound equals the universal lower bound, and the
compactness argument in \cref{thm:certification} ensures attainment.
\end{proof}

Both matching constants are thus determined by the bilateral welfare
problems. The lifting framework supplies the composition; the
endowment charge and the MHR-preserving conditioning specify its
welfare application and its current scope.

\section{Conclusion}\label{sec:conclusion}
The sharp welfare loss is governed by an affine budget comparison that
retains the sellers' initial endowment. For arbitrary priors it reduces
to power-law layers; for MHR buyers it reduces to a typewise allocation
and a matching seller satisfying one boundary equation. Both reductions
preserve enough structure to determine lower and upper bounds together.
The existing bilateral-to-matching framework, augmented by a full-prior
endowment charge, transfers the unrestricted constant under
downward-closed feasibility. MHR-preserving rectangle conditioning
transfers the MHR constant to ordinary matching markets.

Several questions remain. It is natural to seek analytic evaluations of
the two variational minima, beyond their exact characterizations and
finite interval proofs. Extending the sharp MHR welfare guarantee to
additional downward-closed feasibility constraints remains open here;
it requires local rules beyond those supplied by the unrestricted
buyer-prior theorem or a replacement for rectangle packing. Ex post IR, payments only at trade, and
correlated types also change steps used here. Finally, uniqueness in
this paper concerns the normalized \revise{canonical}{reduced-form} parameter triples; it
does not classify all extremal distributions or all optimal mechanisms.

\section*{Declaration of Generative AI Use}
During the preparation of this manuscript, the authors used generative AI tools to assist with language editing, exposition, and manuscript preparation. All mathematical arguments, proofs, results, and conclusions were reviewed and independently verified by the authors, who take full responsibility for the content of the manuscript.

\printbibliography[heading=bibintoc]

@article{MS83,
  author       = {Myerson, Roger B. and Satterthwaite, Mark A.},
  title        = {Efficient Mechanisms for Bilateral Trading},
  journaltitle = {Journal of Economic Theory},
  year         = {1983},
  volume       = {29},
  number       = {2},
  pages        = {265--281}
}

@inproceedings{KPV22,
  author       = {Kang, Zi Yang and Pernice, Francisco and Vondr{\'a}k, Jan},
  title        = {Fixed-Price Approximations in Bilateral Trade},
  booktitle    = {Proceedings of the 2022 ACM-SIAM Symposium on Discrete Algorithms},
  year         = {2022},
  pages        = {2964--2985},
  publisher    = {SIAM}
}

@inproceedings{CW23,
  author       = {Cai, Yang and Wu, Jinzhao},
  title        = {On the Optimal Fixed-Price Mechanism in Bilateral Trade},
  booktitle    = {Proceedings of the 55th Annual ACM Symposium on Theory of Computing},
  year         = {2023},
  pages        = {737--750},
  publisher    = {ACM},
  eprint       = {2301.05167},
  eprinttype   = {arxiv},
  eprintclass  = {cs.GT}
}

@inproceedings{LRW23,
  author       = {Liu, Zhengyang and Ren, Zeyu and Wang, Zihe},
  title        = {Improved Approximation Ratios of Fixed-Price Mechanisms in Bilateral Trades},
  booktitle    = {Proceedings of the 55th Annual ACM Symposium on Theory of Computing},
  year         = {2023},
  pages        = {751--760},
  publisher    = {ACM},
  doi          = {10.1145/3564246.3585160},
  eprint       = {2303.15711},
  eprinttype   = {arxiv},
  eprintclass  = {cs.GT}
}

@article{GK26,
  author       = {Giambartolomei, Giordano and de Keijzer, Bart},
  title        = {On the Approximation Ratio of Optimal Fixed-Price Mechanisms for Single and Multi-Unit Bilateral Trade},
  journaltitle = {Proceedings of the AAAI Conference on Artificial Intelligence},
  year         = {2026},
  volume       = {40},
  number       = {20},
  pages        = {16946--16953},
  doi          = {10.1609/aaai.v40i20.38741}
}

@online{JGC26,
  author       = {Jiang, Tao and Gao, Minbo and Cai, Shaowei},
  title        = {The Exact Approximation Ratio of the Optimal Fixed-Price Mechanism in Bilateral Trade},
  year         = {2026},
  eprint       = {2609.27878},
  eprinttype   = {arxiv},
  eprintclass  = {cs.DS}
}

@online{DS26,
  author       = {Dobzinski, Shahar and Shaulker, Ariel},
  title        = {Welfare Maximization in Bilateral Trade: Improved Approximation Guarantees Beyond the Fixed Price Barrier},
  year         = {2026},
  eprint       = {2606.04890},
  eprinttype   = {arxiv},
  eprintclass  = {cs.GT}
}

@online{LQRW26,
  author       = {Liu, Zhengyang and Qin, Ying and Ren, Zeyu and Wang, Zihe},
  title        = {Second-Best Bilateral Trade Is {$1/2$} Efficient},
  year         = {2026},
  eprint       = {2606.03849},
  eprinttype   = {arxiv},
  eprintclass  = {cs.GT},
  version      = {2}
}

@online{BLWZ26,
  author       = {Bei, Xiaohui and Li, Bo and Wu, Wenhao and Zhou, Shengwei},
  title        = {Second-Best Gains from Trade in Matching Markets},
  year         = {2026},
  eprint       = {2609.18724},
  eprinttype   = {arxiv},
  eprintclass  = {cs.GT}
}

@online{LQW26,
  author       = {Liu, Zhengyang and Qin, Ying and Wang, Zihe},
  title        = {From Bilateral Trade to Matching Markets: Sharp Gains from Trade},
  year         = {2026},
  eprint       = {2609.30702},
  version      = {1},
  eprinttype   = {arxiv},
  eprintclass  = {cs.GT}
}

@online{mpmath,
  author       = {{The mpmath developers}},
  shorthand    = {mpmath},
  title        = {mpmath 1.3.0 Documentation: Arbitrary-Precision Interval Arithmetic},
  url          = {https://mpmath.org/doc/current/contexts.html},
  urldate      = {2026-09-28}
}

@article{HR87,
  author       = {Hagerty, Kathleen M. and Rogerson, William P.},
  title        = {Robust Trading Mechanisms},
  journaltitle = {Journal of Economic Theory},
  year         = {1987},
  volume       = {42},
  number       = {1},
  pages        = {94--107},
  doi          = {10.1016/0022-0531(87)90104-9}
}

@article{BD21,
  author       = {Blumrosen, Liad and Dobzinski, Shahar},
  title        = {(Almost) Efficient Mechanisms for Bilateral Trading},
  journaltitle = {Games and Economic Behavior},
  year         = {2021},
  volume       = {130},
  pages        = {369--383},
  doi          = {10.1016/j.geb.2021.08.011},
  eprint       = {1604.04876},
  eprinttype   = {arxiv},
  eprintclass  = {cs.GT}
}

@article{McA08,
  author       = {McAfee, R. Preston},
  title        = {The Gains from Trade under Fixed Price Mechanisms},
  journaltitle = {Applied Economics Research Bulletin},
  year         = {2008},
  volume       = {1},
  number       = {1},
  pages        = {1--10},
  url          = {https://mc4f.ee/Papers/PDF/GFTConstantPrice.pdf}
}

@inproceedings{BM16,
  author       = {Blumrosen, Liad and Mizrahi, Yehonatan},
  title        = {Approximating Gains-from-Trade in Bilateral Trading},
  booktitle    = {Web and Internet Economics: 12th International Conference},
  series       = {Lecture Notes in Computer Science},
  volume       = {10123},
  year         = {2016},
  pages        = {400--413},
  publisher    = {Springer},
  doi          = {10.1007/978-3-662-54110-4_28}
}

@inproceedings{CBKLT16,
  author       = {Colini-Baldeschi, Riccardo and de Keijzer, Bart and Leonardi, Stefano and Turchetta, Stefano},
  title        = {Approximately Efficient Double Auctions with Strong Budget Balance},
  booktitle    = {Proceedings of the Twenty-Seventh Annual ACM-SIAM Symposium on Discrete Algorithms},
  year         = {2016},
  pages        = {1424--1443},
  publisher    = {SIAM},
  doi          = {10.1137/1.9781611974331.ch98}
}

@article{CBGKRT20,
  author       = {Colini-Baldeschi, Riccardo and Goldberg, Paul W. and de Keijzer, Bart and Leonardi, Stefano and Roughgarden, Tim and Turchetta, Stefano},
  title        = {Approximately Efficient Two-Sided Combinatorial Auctions},
  journaltitle = {ACM Transactions on Economics and Computation},
  year         = {2020},
  volume       = {8},
  number       = {1},
  eid          = {4},
  pagetotal    = {29},
  doi          = {10.1145/3381523},
  eprint       = {1611.05342},
  eprinttype   = {arxiv},
  eprintclass  = {cs.GT}
}

@inproceedings{DFL21,
  author       = {D{\"u}tting, Paul and Fusco, Federico and Lazos, Philip and Leonardi, Stefano and Reiffenh{\"a}user, Rebecca},
  title        = {Efficient Two-Sided Markets with Limited Information},
  booktitle    = {Proceedings of the 53rd Annual ACM SIGACT Symposium on Theory of Computing},
  year         = {2021},
  pages        = {1452--1465},
  publisher    = {ACM},
  doi          = {10.1145/3406325.3451076},
  eprint       = {2003.07503},
  eprinttype   = {arxiv},
  eprintclass  = {cs.GT}
}

@inproceedings{BCWZ17,
  author       = {Brustle, Johannes and Cai, Yang and Wu, Fa and Zhao, Mingfei},
  title        = {Approximating Gains from Trade in Two-Sided Markets via Simple Mechanisms},
  booktitle    = {Proceedings of the 18th ACM Conference on Economics and Computation},
  year         = {2017},
  pages        = {589--590},
  publisher    = {ACM},
  eprint       = {1706.04637},
  eprinttype   = {arxiv},
  eprintclass  = {cs.GT}
}

@inproceedings{DMSW22,
  author       = {Deng, Yuan and Mao, Jieming and Sivan, Balasubramanian and Wang, Kangning},
  title        = {Approximately Efficient Bilateral Trade},
  booktitle    = {Proceedings of the 54th Annual ACM SIGACT Symposium on Theory of Computing},
  year         = {2022},
  pages        = {718--721},
  publisher    = {ACM},
  eprint       = {2111.03611},
  eprinttype   = {arxiv},
  eprintclass  = {cs.GT}
}

@inproceedings{Fei22,
  author       = {Fei, Yumou},
  title        = {Improved Approximation to First-Best Gains-from-Trade},
  booktitle    = {Web and Internet Economics: 18th International Conference},
  year         = {2022},
  pages        = {204--218},
  publisher    = {Springer},
  eprint       = {2205.00140},
  eprinttype   = {arxiv},
  eprintclass  = {cs.GT}
}

@inproceedings{CGMZ21,
  author       = {Cai, Yang and Goldner, Kira and Ma, Steven and Zhao, Mingfei},
  title        = {On Multi-Dimensional Gains from Trade Maximization},
  booktitle    = {Proceedings of the 2021 ACM-SIAM Symposium on Discrete Algorithms},
  year         = {2021},
  pages        = {1079--1098},
  publisher    = {SIAM},
  eprint       = {2007.13934},
  eprinttype   = {arxiv},
  eprintclass  = {cs.GT}
}

@inproceedings{DMSWW25,
  author       = {Deng, Yuan and Mao, Jieming and Sivan, Balasubramanian and Wang, Kangning and Wu, Jinzhao},
  title        = {Approximately Efficient Bilateral Trade with Samples},
  booktitle    = {Proceedings of the 26th ACM Conference on Economics and Computation},
  year         = {2025},
  pages        = {206--223},
  publisher    = {ACM},
  doi          = {10.1145/3736252.3742519},
  eprint       = {2502.13122},
  eprinttype   = {arxiv},
  eprintclass  = {cs.GT}
}

@inproceedings{DS24,
  author       = {Dobzinski, Shahar and Shaulker, Ariel},
  title        = {Bilateral Trade with Correlated Values},
  booktitle    = {Proceedings of the 56th Annual ACM Symposium on Theory of Computing},
  year         = {2024},
  pages        = {237--246},
  publisher    = {ACM},
  doi          = {10.1145/3618260.3649659},
  eprint       = {2308.09964},
  eprinttype   = {arxiv},
  eprintclass  = {cs.GT}
}

@inproceedings{DEGST25,
  author       = {Dobzinski, Shahar and Eden, Alon and Goldner, Kira and Shaulker, Ariel and Tsilivis, Thodoris},
  title        = {Bilateral Trade with Interdependent Values: Information vs. Approximation},
  booktitle    = {Proceedings of the 26th ACM Conference on Economics and Computation},
  year         = {2025},
  pages        = {641--665},
  publisher    = {ACM},
  doi          = {10.1145/3736252.3742607},
  eprint       = {2506.23896},
  eprinttype   = {arxiv},
  eprintclass  = {cs.GT}
}

@inproceedings{BRTW26,
  author       = {Babaioff, Moshe and Rubinstein, Aviad and Tan, Xizhi and Wang, Kangning},
  title        = {Approximating Gains-from-Trade in Matching Markets},
  booktitle    = {Proceedings of the 58th Annual ACM Symposium on Theory of Computing},
  year         = {2026},
  pages        = {710--721},
  publisher    = {ACM},
  doi          = {10.1145/3798129.3800786},
  eprint       = {2604.00129},
  eprinttype   = {arxiv},
  eprintclass  = {cs.GT}
}
\clearpage
\appendix
\crefalias{section}{appendix}

\section{General Priors and Payment Conventions}\label{sec:extension}
The reductions and transfer proofs first use finite or bounded priors.
We now justify the extensions invoked in the main theorem, tracking
the endowment alongside gains and retaining the MHR restriction.
For unrestricted priors the arguments allow downward-closed matching
feasibility; the MHR guarantees use ordinary matching markets.

\subsection{Bounded Priors}\label{u:lem:rounding}
Round each buyer down and each seller up to a grid of mesh $\delta$. An
allocation defined on the finite rounded supports is extended using the
greatest supported buyer grid point below the report and the least supported
seller grid point above the report. If that point is absent, that agent is
unserved; endpoint service is extended constantly on the other side.
For the true types these maps agree almost surely with rounding.

Use buyer and seller minimal-rent envelopes on the resulting step functions.
On an interval with constant buyer service $\revise{X}{x_B}$, the term $b\revise{X}{x_B}$ and the integral
of $\revise{X}{x_B}$ vary by the same amount, so the payment is constant and equals the
finite payment at its lower supported endpoint. For a seller interval the
terms $s\revise{Y}{x_S}$ and $\int_s\revise{Y}{x_S}$ cancel at its upper supported endpoint. Across an
empty grid cell the next supported endpoint and the same identity give the
finite adjacent utility increment. Thus interim payments at true types,
their expectations, BIC, interim IR, and budget feasibility are preserved.
Actual gains weakly increase relative to rounded gains.

If $m$ is the number of sellers and $d$ the maximum matching size, then
let $A_\delta$ and $G_\delta$ denote the expected seller endowment and first-best gains from trade in the rounded instance. Then
$A\le A_\delta\le A+m\delta$ and $G_\delta\ge G-2d\delta$.
The extended allocation's \revise{attainable}{feasible} payment surplus is at least its
preserved surplus. An affine bound for rounded priors therefore gives
\[
 V_{\revise{\lambda}{\alpha}}+\kappa A\ge\beta G-(2\beta d+\kappa m)\delta.
\]
The same argument for a budget-feasible ratio gives
$G_{\SB}+A\ge\rho(G+A)
-[2\rho d+(1-\rho)m]\delta$.
Let $\delta\downarrow0$. This proves the unrestricted extensions without
postulating strong duality for arbitrary Borel priors.

For MHR buyers, keep every buyer unchanged and round only seller
costs upward. The same constant-service envelope calculation preserves
seller payments; buyer interim allocations and payments are unchanged.
Expected budget feasibility is therefore preserved, while actual gains
weakly increase. This applies the preceding argument without leaving
the MHR class.

For $m$ sellers, $A\le A_\delta\le A+m\delta$, and efficient gains satisfy $G_\delta\ge G-m\delta$. Thus a fixed-multiplier affine guarantee on rounded priors yields
\[
 V_{\revise{\lambda}{\alpha}}+\kappa A\ge\beta G-(\beta+\kappa)m\delta.
\]
The budget-feasible welfare guarantees have vanishing errors as well.
Letting $\delta\downarrow0$ proves the MHR extensions without invoking
duality for arbitrary Borel sellers.

\subsection{Unbounded Supports}
Clip every active type at $K$. Clipping preserves the unrestricted class, and a clipped MHR buyer remains MHR. Nonnegative edge gains satisfy
\[
 (\min(b,K)-\min(s,K))_+\uparrow(b-s)_+,
\]
so the clipped first-best gains $G_K\uparrow G$ and seller baseline $A_K\uparrow A$. There are finitely many matchings, which justifies exchanging this monotone limit with their maximum and expectation.

Use the positive-gain budget-feasible guarantees obtained from the restricted compact allocation class in \cref{lem:duality}. A seller of clipped cost $K$ is never served. Extend the mechanism to original reports by clipping: buyer service is constant above $K$, and seller service is zero there. Envelope payments remain unchanged, incentives and expected budget are preserved, and gains weakly increase. In the original market the welfare is at least
\[
 A+\rho(A_K+G_K)-A_K-o(1).
\]
Let the approximation error vanish, then $K\to\infty$. This proves the welfare lower bounds whenever $0<W_{\FB}<\infty$. A nonnegative MHR buyer has finite mean because its concave $g$ is bounded above by an affine function for large arguments. Isolated buyers can be removed. Every remaining buyer has finite expectation when $W_{\FB}<\infty$, because the first-best welfare dominates that buyer's value by considering one incident edge. No attainment claim is made for arbitrary unbounded priors.

\subsection{Atomless Buyers and Balanced Transfers}
For a capped buyer with cap length $T>0$, replace $g(t)=\ell+\min(t,T)$ by
\[
 g_\varepsilon(t)=\begin{cases}
 \ell+t,&t\le T,\\
 \ell+T+\varepsilon(1-e^{-(t-T)}),&t>T,
 \end{cases}\qquad0<\varepsilon\le1.
\]
This function is increasing and concave, with derivative dropping from $1$ to $\varepsilon$ and then decreasing. It gives an atomless MHR buyer. Values converge uniformly, and normalized scores $g_\varepsilon-ag_\varepsilon'$ converge in $L^1$ to the capped score. For every fixed hard seller with a monotone normalized score, efficient gains and the exact score expression for $V_{\revise{\lambda}{\alpha}}$ converge. Thus the strict upper bounds and the tight infima persist in the atomless class. 

Finally, payment balance follows from the independent-types,
signed-transfer conversion of
\cite[Appendix~A.1]{LQW26}, which applies to every mechanism above. Its formula
uses only integrable expectations, so the finite-type argument extends
unchanged to the independent Borel priors considered here.
After rebating the nonnegative expected surplus, it reconstructs
profilewise balanced transfers with the same interim payments.
Allocations, BIC, and interim IR are preserved, so all welfare guarantees
and sharpness statements hold under pointwise SBB as well. This
conversion does not assert ex post IR or payments only at trade.

These extensions complete the passage from bounded instances to the
prior and payment conventions of the main theorems. The next two
appendices verify the unrestricted and MHR constants, respectively.

\section{\texorpdfstring{Proof of Theorem~\ref*{u:thm:scalar}}{Proof of the unrestricted welfare theorem}}\label{app:unrestricted}
This appendix proves \cref{u:thm:scalar}. We give the certificate
structure, justify its global and local bounds, and collect the
derivatives needed by the verifier. All finite inputs and replay
commands are described in \cref{sec:reproduction}.

\subsection{Certificate Structure}\label{u:sec:scalar}
The certificate first confines every candidate minimum to a compact
box, then uses local convexity to prove sharpness and uniqueness.
We state the verification procedure before supplying its analytic bounds.

\subsubsection{Verification Algorithm}
Write $z=(a,p,r)$. Since $D(\rho,z)=(A+I)(R_{\mathrm U}(z)-\rho)$, it suffices to certify signs of $D$. We use the hierarchy
\[
 0.8882516<\rho_{\mathrm{lo}}<\rho_{\mathrm{hi}}<\overline\rho=0.8882517,
 \qquad K\subset\operatorname{int}K_0.
\]
Here $\rho_{\mathrm{lo}},\rho_{\mathrm{hi}}$ are the exact decimal endpoints in \cref{u:subsec:point}; $K_0$ is an outer box containing every point with $R_{\mathrm U}\le\overline\rho$; and $K$ is the cube of radius $0.012$ around a rational point $c$. The local Hessian bound will hold for every $\rho\in[0.8882516,\overline\rho]$, so it applies simultaneously to both endpoint tests and the boundary test at $\overline\rho$.

The finite certificate consists of these rational parameters, a subdivision tree for $K_0$, labels for its leaves, and a box cover of $K$. A leaf is labeled either as contained in $K$, as having positive $D(\overline\rho,\cdot)$, or as having a strictly signed $p$- or $r$-partial derivative. \Cref{u:alg:verification} checks these assertions with interval arithmetic. A failed interval test rejects the proposed certificate. Thus termination of verification depends only on the finite input, not on convergence of numerical optimization.

\begin{algorithm}[H]
\caption{Verifying the scalar minimum}\label{u:alg:verification}
\textbf{Input:} Rational parameters $\rho_{\mathrm{lo}},\rho_{\mathrm{hi}},\overline\rho,c,K,K_0$, a labeled subdivision of $K_0$, and a local box cover of $K$.
\begin{enumerate}[label=\arabic*.,leftmargin=2em]
\item Check the parameter hierarchy, $K\subset\operatorname{int}K_0$, and exact coverage of the global subdivision and the local cover.
\item Verify the exterior inequalities proving $D(\overline\rho,z)>0$ outside $\operatorname{int}K_0$. For each global leaf, verify its claimed containment, positive value bound, or nonzero partial derivative.
\item On every local box, verify $\nabla^2D(\rho,z)\succeq m\revise{I}{I_3}$ for $\rho\in[0.8882516,\overline\rho]$, with $m=0.01$, where $I_3$ denotes the $3\times3$ identity matrix.
\item Verify $D(\rho_{\mathrm{lo}},c)-\|\nabla D(\rho_{\mathrm{lo}},c)\|_2^2/(2m)>0$, $D(\rho_{\mathrm{hi}},c)<0$, and the strong-convexity bound proving $D(\overline\rho,z)>0$ on $\partial K$.
\item Accept the claims $\rho_{\mathrm{lo}}<\min R_{\mathrm U}<\rho_{\mathrm{hi}}$ and uniqueness of the minimizer in $K$ if every check passes; otherwise reject.
\end{enumerate}
The test formulas and finite data are specified in \cref{u:sec:verification}.
\end{algorithm}

We obtain the global subdivision by bisecting a longest side of each unresolved box until one of its labels can be proposed. The subsequent verifier checks the resulting finite object independently. All numerical errors are enclosed by outward-rounded operations and explicit series remainders; no approximation slack is added to the final welfare inequality.

The value and derivative tests exclude nonlocal minima only after the
two boundary checks succeed. \Cref{u:lem:cover-logic} formalizes this
step; the proof of \cref{u:thm:scalar} in \cref{u:subsec:point} combines
it with the local convexity and point bounds. We give that argument
once, after verifying its inputs.

\paragraph{Complexity.}
Let $L$ be the number of leaves in the global subdivision, $H$ the number of local Hessian boxes, $N$ the number of terms used to enclose each beta integral, $Q$ the number of quadrature intervals per local integral, and $M$ the series length for the final point evaluation. Verification takes $O(LN+HQ+M)$ interval arithmetic operations. The tree and coverage checks take $O(L+H)$ additional operations, and the certificate can be checked using $O(L+H+N+M)$ space. These are bounds in the size of the finite certificate; the number of adaptive subdivisions depends on the function and the required margins.

\subsubsection{Verified Bounds}
\Cref{u:tab:certificate} summarizes the finite computation. The global integral bounds truncate a positive series after degree $N=40$ with an explicit geometric remainder. The local Hessian bounds use $Q=256$ intervals; convexity of the integrands makes midpoint sums lower bounds and trapezoid sums upper bounds. The point evaluation uses $M=700$ terms. All formulas and remainder estimates are proved in \cref{u:sec:verification,u:app:derivatives}. The accompanying programs implement the same tests using mpmath interval arithmetic and an independent MPFR implementation.

\begin{table}[t]
\centering
\begin{tabular}{lrl}
\toprule
Verification & Number & Certified conclusion\\
\midrule
Exterior inequalities & $6$ & $D(\overline\rho,\cdot)>0$ outside $\operatorname{int}K_0$\\
Positive-value boxes & $74{,}414$ & $D(\overline\rho,\cdot)>0$\\
Signed-partial boxes & $92{,}091$ & $D_p\ne0$ or $D_r\ne0$\\
Boxes contained in $K$ & $1{,}000$ & covered by local analysis\\
Local Hessian boxes & $512$ & $\nabla^2D\succeq0.01I$\\
\bottomrule
\end{tabular}
\caption{The finite certificate. The global subdivision has $167{,}505$ leaves. The least signed margin in its nonlocal tests exceeds $1.62\cdot10^{-8}$. The final point and boundary inequalities are given in \cref{u:eq:lower-point,u:eq:upper-point,u:eq:boundary-point}.}
\label{u:tab:certificate}
\end{table}

The table records the accepted finite certificate. We next prove
the inequalities used by its exterior, integral, and local tests.

\subsection{Enclosure Bounds}\label{u:sec:verification}
Throughout this subsection, $\overline\rho=0.8882517$ and $C=1-\overline\rho$. We give the inequalities and integration bounds used in the finite certificate. Interval arithmetic encloses each elementary operation and every series remainder.

\subsubsection{Boundary Exclusion}\label{u:subsec:compact}
Let $\revise{\alpha_-}{a_-},\revise{\pi_-}{p_-},\revise{\eta_-}{r_-},\revise{\alpha_+}{a_+},\revise{\pi_+}{p_+}$ be the dyadic rational endpoints stored in the certificate, where $a_-,a_+$ bound $a$, $p_-,p_+$ bound $p$, and $r_-$ is the lower bound for $r$; with approximate values
\[
 0.1,\quad0.02,\quad0.00005,\quad0.999999,\quad0.9.
\]
Each is the nearest binary64 number to the displayed decimal. This convention specifies exact rational endpoints. Define
\begin{equation}\label{u:eq:rootbox}
 K_0=[\revise{\alpha_-}{a_-},\revise{\alpha_+}{a_+}]\times[\revise{\pi_-}{p_-},\revise{\pi_+}{p_+}]\times[\revise{\eta_-}{r_-},1].
\end{equation}

\begin{lemma}[Exterior inequalities]\label{u:lem:exterior}
At $\rho=\overline\rho$, $D>0$ outside the interior of $K_0$, for all $(a,p,r)\in(0,1)^3$; it is also positive on the face $r=1$ of $K_0$.
\end{lemma}
\begin{proof}[Proof of \cref{u:lem:exterior}]
The seller mean decreases with $a$ and $p$, as is clear from \cref{u:eq:layer-quantiles}. Sending $a\uparrow1$ gives
\begin{equation}\label{u:eq:elementary-bounds}
\begin{split}
 A_a(p)&\ge1-p+p\log p,\\
 I_a(p,r)&\le1-A_a(p),\\
 I_a(p,r)&\le\E B=\frac{r-a r^{1/a}}{1-a}
                 \le r(1-\log r),\\
 I_a(p,r)&\le pr\,\mathrm B(1-a,1-a),
\end{split}
\end{equation}
where $\mathrm B(k,k)=\int_0^1[t(1-t)]^{k-1}\dd t$. The buyer-mean inequality follows from its quantile increasing with $a$, and the other bounds follow directly from \cref{u:eq:I}.

For small $a$, write $K(a)=(1-a)\mathrm B(1-a,1-a)$. The beta-gamma identity gives
\[
 K(a)=\frac{2\Gamma(2-a)^2}{\Gamma(3-2a)}.
\]
The function $K$ increases with $a$: with $k=1-a$, its logarithmic derivative in $k$ is $2\psi(1+k)-2\psi(1+2k)<0$. Here strict increase of $\psi=(\log\Gamma)'$ follows from strict log-convexity of the gamma integral, which itself is H\"older's inequality. Also $K(a)\ge K(0)=1$. Therefore \cref{u:eq:elementary-bounds} implies
\[
 R_{\mathrm U}(a,p,r)\ge\min\{1,K(a)^{-1}\}=K(a)^{-1}.
\]
Thus $a\le\revise{\alpha_-}{a_-}$ is excluded by the first of the interval tests below.

For $p\le\revise{\pi_-}{p_-}$, we have
\[
 D\ge A_a(p)-\overline\rho
 \ge C-\revise{\pi_-}{p_-}(1-\log\revise{\pi_-}{p_-})>0.
\]
For $p\ge\revise{\pi_+}{p_+}$, use $I\le\E B\le r/(1-a)$ to obtain
\[
 D\ge C A_a(p)+(p-\overline\rho)r/(1-a)>0.
\]
Once $p\le\revise{\pi_+}{p_+}$, the region $r\le\revise{\eta_-}{r_-}$ is excluded by
\[
 D\ge C(1-\revise{\pi_+}{p_+}+\revise{\pi_+}{p_+}\log\revise{\pi_+}{p_+})
       -\overline\rho\,\revise{\eta_-}{r_-}(1-\log\revise{\eta_-}{r_-})>0.
\]
With $p\ge\revise{\pi_-}{p_-}$ and $r\ge\revise{\eta_-}{r_-}$, the region $a\ge\revise{\alpha_+}{a_+}$ is excluded by
\[
 D\ge\frac{\revise{\pi_-}{p_-}\revise{\eta_-}{r_-}}{1-\revise{\alpha_+}{a_+}}-\overline\rho>0.
\]
Finally, at $r=1$,
\[
 D=A_a(p)+\frac{p}{1-a}-\overline\rho\ge C>0.
\]
All bounds include the relevant equality faces. Interval evaluation gives the following positive lower bounds:
\begin{center}
\begin{tabular}{lr}
\toprule
Test & Certified lower bound\\
\midrule
$K(\revise{\alpha_-}{a_-})^{-1}-\overline\rho$ & $0.0179$\\
$C-\revise{\pi_-}{p_-}(1-\log\revise{\pi_-}{p_-})$ & $0.0135$\\
$\revise{\pi_+}{p_+}-\overline\rho$ & $0.0117$\\
$C(1-\revise{\pi_+}{p_+}+\revise{\pi_+}{p_+}\log\revise{\pi_+}{p_+})-\overline\rho\revise{\eta_-}{r_-}(1-\log\revise{\eta_-}{r_-})$ & $0.000094$\\
$\revise{\pi_-}{p_-}\revise{\eta_-}{r_-}/(1-\revise{\alpha_+}{a_+})-\overline\rho$ & $0.1117$\\
$C$ & $0.1117$\\
\bottomrule
\end{tabular}
\end{center}
The beta value in the first test is evaluated by the elementary series in \cref{u:subsec:series}, not by evaluating $\Gamma$.
\end{proof}

\subsubsection{Finite Global Coverage}\label{u:subsec:cover}
For a parameter box $\revise{B}{\mathcal B_{\rm par}}=[a_-,a_+]\times[p_-,p_+]\times[r_-,r_+]$, the following is a valid lower bound for $D/(pr)$:
\begin{equation}\label{u:eq:box-value}
 C\frac{A_{a_+}(p_+)}{p_+r_+}
 +\frac1{1-a_-}
 -\overline\rho\frac{I_{a_+}(p_-,r_-)}{p_-r_-}.
\end{equation}
Indeed, $A_a(p)$ decreases in $a,p$, while
\[
 \frac{I_a(p,r)}{pr}
 =\int_0^1\min\{p^{-1},(1-t)^{-a}\}
              \min\{r^{-1},t^{-a}\}\dd t
\]
increases in $a$ and decreases in $p,r$.

For derivative tests define
\begin{align}
 U(a,r,z)&=\int_0^z(1-t)^{-a}\min\{1,r t^{-a}\}\dd t,\label{u:eq:U}\\
 V(a,p,z)&=\int_z^1t^{-a}\min\{1,p(1-t)^{-a}\}\dd t.\label{u:eq:Vaux}
\end{align}
Differentiation under the integral, with matching boundary values at the moving cutoffs, gives
\begin{equation}\label{u:eq:first-partials}
 I_p=U(a,r,1-p^{1/a}),\qquad I_r=V(a,p,r^{1/a}),\qquad
 -A_p=\frac{1-p^{(1-a)/a}}{1-a}.
\end{equation}
Both integrands increase with $a$ and the other displayed mass parameter. The function $U$ increases with its last argument, whereas $V$ decreases with it. Set
\[
 \ell_- =\frac{1-p_+^{(1-a_+)/a_+}}{1-a_-},\qquad
 \ell_+ =\frac{1-p_-^{(1-a_-)/a_-}}{1-a_+}.
\]
Then $\ell_-\le-A_p\le\ell_+$. Uniform derivative bounds on $\revise{B}{\mathcal B_{\rm par}}$ are
\begin{align}
 D_p&\ge-C\ell_++\frac{r_-}{1-a_-}
       -\overline\rho\,U(a_+,r_+,1-p_-^{1/a_-}),\label{u:eq:dp-lower}\\
 D_p&\le-C\ell_-+\frac{r_+}{1-a_+}
       -\overline\rho\,U(a_-,r_-,1-p_+^{1/a_+}),\label{u:eq:dp-upper}\\
 D_r&\ge\frac{p_-}{1-a_-}
       -\overline\rho\,V(a_+,p_+,r_-^{1/a_-}),\label{u:eq:dr-lower}\\
 D_r&\le\frac{p_+}{1-a_+}
       -\overline\rho\,V(a_-,p_-,r_+^{1/a_+}).\label{u:eq:dr-upper}
\end{align}
These formulas are valid even where the two saturation intervals overlap.

Let $c=(c_a,c_p,c_r)$ be the following exact rational point:
\begin{equation}\label{u:eq:rational-center}
\begin{split}
 c_a&=0.323144212181945529204914052954389108756940743,\\
 c_p&=0.654479031166666393676935731392503971225432437,\\
 c_r&=0.325165829719811989207067293084997587516390811.
\end{split}
\end{equation}
Define the rational local cube
\begin{equation}\label{u:eq:localbox}
 K=\prod_{j\in\{a,p,r\}}[c_j-0.012,c_j+0.012].
\end{equation}
We bisect a longest side of any unresolved box, breaking ties by coordinate order. Each leaf is contained in $K$, satisfies the positive lower bound in \cref{u:eq:box-value}, or satisfies one of the signed derivative bounds in \cref{u:eq:dp-lower,u:eq:dp-upper,u:eq:dr-lower,u:eq:dr-upper}. An exact check of the subdivision tree verifies coverage and local containment. The leaf counts and numerical margins are summarized in \cref{u:tab:certificate}.

\begin{lemma}[Excluding nonlocal minima]\label{u:lem:cover-logic}
Suppose $D(\overline\rho,\cdot)>0$ on $\partial K$ and on the boundary of $K_0$. The verified cover then implies
\[
 D(\overline\rho,a,p,r)>0\quad\text{on }K_0\setminus\operatorname{int}K.
\]
\end{lemma}
\begin{proof}[Proof of \cref{u:lem:cover-logic}]
Otherwise the continuous function has a nonpositive minimum on the compact set $K_0\setminus\operatorname{int}K$. Boundary positivity places this minimum in the open set $\operatorname{int}K_0\setminus K$, where its gradient must vanish. The function is continuously differentiable there, including across the saturation-overlap surface, as follows directly from \cref{u:eq:I} and matching cutoff values. A leaf containing this minimum cannot be a local leaf. It cannot be a positive-value leaf, and it cannot be a strictly signed-partial leaf. This contradicts the cover.
\end{proof}

\subsubsection{Integral Enclosures}\label{u:subsec:series}
For $k=1-a$ and $0\le z<1$, the incomplete beta primitive is
\begin{equation}\label{u:eq:beta-series}
 \mathrm B_z(k,k)=z^k\sum_{n=0}^\infty c_n z^n,
 \qquad c_n=\frac{(a)_n}{n!(n+k)}.
\end{equation}
This follows by expanding $(1-t)^{-a}$ in its positive binomial series and integrating termwise. The coefficients satisfy
\begin{equation}\label{u:eq:beta-recurrence}
 c_0=1/k,\qquad
 c_n=c_{n-1}\frac{(a+n-1)(k+n-1)}{n(k+n)},\qquad
 0<c_n\le c_{n-1}.
\end{equation}
After term $N$, the remaining part is at most
\begin{equation}\label{u:eq:beta-tail}
 z^k c_N\frac{z^{N+1}}{1-z}.
\end{equation}
For $z>1/2$, use symmetry
\[
 \mathrm B_z(k,k)=\mathrm B(k,k)-\mathrm B_{1-z}(k,k),\qquad
 \mathrm B(k,k)=2\mathrm B_{1/2}(k,k).
\]
Thus all complete and incomplete beta values use only the same elementary positive series. The global verification truncates after $N=40$, with every remainder included as a nonnegative interval.

Let $x_0=p^{1/a}$, $y_0=r^{1/a}$, and
\begin{equation}\label{u:eq:T0}
 T_0(a,z)=\int_0^z(1-t)^{-a}\dd t
         =\frac{1-(1-z)^{1-a}}{1-a}.
\end{equation}
In this subsection write $I:=I_a(p,r)$. If $x_0+y_0\le1$, splitting the integral at its two cutoffs gives
\begin{equation}\label{u:eq:I-split}
 I=pT_0(a,y_0)+rT_0(a,x_0)
       +pr\bigl[\mathrm B_{1-x_0}(k,k)-\mathrm B_{y_0}(k,k)\bigr].
\end{equation}
If $x_0+y_0\ge1$, every realized buyer value is at least every seller value, so
\begin{equation}\label{u:eq:I-overlap}
 I=\E B-\E S=\frac{r-a y_0}{k}-A_a(p).
\end{equation}
The formulas agree at equality. If an interval calculation cannot distinguish the cases, the hull of the two formula enclosures is safe. The integrals $U,V$ are evaluated by the same primitives, splitting at their saturation point and replacing an integral over a reversed interval by zero.

\subsubsection{Local Strong Convexity}\label{u:subsec:convexity}
For $z\in[0,1)$ define
\begin{equation}\label{u:eq:Tm}
 T_m(a,z)=\int_0^z(1-t)^{-a}[-\log(1-t)]^m\dd t.
\end{equation}
With $w=-(1-a)\log(1-z)$,
\begin{equation}\label{u:eq:Tm-explicit}
 T_m(a,z)=\frac{m!}{(1-a)^{m+1}}
     \left(1-e^{-w}\sum_{j=0}^m\frac{w^j}{j!}\right).
\end{equation}
On the entire local cube, $x_0+y_0<1$. Put
\begin{equation}\label{u:eq:Jm}
 J_m=\int_{y_0}^{1-x_0}
       [t(1-t)]^{-a}[-\log(t(1-t))]^m\dd t,\quad m=0,1,2.
\end{equation}
Analytic first and second derivatives are listed and derived in \cref{u:app:derivatives}. They express the Hessian of $D$ in terms of $T_m,J_m$ and elementary operations.

The integrand in $J_m$ is convex. To see this, write $L(t)=-\log(t(1-t))$. Then $L\ge\log4>1$, and
$L''-(L')^2=2/[t(1-t)]>0$.
Thus $L$ and $\log L$ are convex. The logarithm of the integrand is $aL+m\log L$, which is convex for $a>0$ and $m\ge0$; hence the integrand itself is convex. On every integration interval, the composite midpoint rule is therefore a lower bound and the composite trapezoid rule an upper bound. The local verification uses $256$ equal subintervals for each such bound. Monotonicity in $a$ and expansion of the integration interval enclose all parameter values in a box, rather than merely its center.

We divide each coordinate interval of $K$ into eight parts and enclose the resulting $512$ boxes by rational boxes. Exact comparisons of adjacent endpoints verify that these boxes cover $K$.

On each box, and simultaneously for
$\rho\in[0.8882516,0.8882517]$, the interval Hessian enclosure $H$ satisfies $H-0.01\revise{I}{I_3}\succ0$. This is checked by the scalar LDL pivots
\begin{align*}
 d_1&=H_{11}-0.01,\\
 d_2&=H_{22}-0.01-H_{12}^2/d_1,\\
 d_3&=H_{33}-0.01-H_{13}^2/d_1
       -(H_{23}-H_{12}H_{13}/d_1)^2/d_2.
\end{align*}
All lower interval endpoints are positive; lower bounds over the $512$ boxes exceed, respectively, $1.558$, $0.420$, and $0.048$. It follows that
\begin{equation}\label{u:eq:strong-convexity}
 \nabla^2 D(\rho,z)\succeq m\revise{I}{I_3},
 \qquad m=0.01,\quad z\in K,\quad
 \rho\in[0.8882516,0.8882517].
\end{equation}
Although the entries of an interval matrix need not vary independently for a true Hessian, the interval LDL calculations enclose every true pivot; positivity of their lower endpoints is sufficient.

\subsubsection{Point Bounds and Uniqueness}\label{u:subsec:point}
For the final enclosure, define the exact decimal rationals
\begin{align*}
 \rho_{\mathrm{lo}}&=0.88825169032566246988432630129334211037,\\
 \rho_{\mathrm{hi}}&=0.88825169032566246988432630129334211038.
\end{align*}
At the rational center $c$ in \cref{u:eq:rational-center}, we evaluate $D$ and its gradient by interval arithmetic. To enclose the $a$-derivative of \cref{u:eq:beta-series} at a fixed endpoint $z$, set
\[
 H_n=\sum_{j=0}^{n-1}\frac1{a+j},\qquad
 L_n=H_n+\frac1{n+k}-\log z.
\]
The derivative series is $z^k\sum c_n z^nL_n$. Since
$L_{N+j}\le L_N+j/(a+N)$ and $c_{N+j}\le c_N$, its tail after term $N$ is at most
\begin{equation}\label{u:eq:derivative-tail}
 z^k c_N z^N\left(
 L_N\frac z{1-z}+\frac z{(a+N)(1-z)^2}\right).
\end{equation}
Every term is nonnegative. We use $N=700$ and these explicit remainder bounds to evaluate the gradient.

The resulting validated inequalities, rounded here to weaker bounds, are
\begin{align}
 D(\rho_{\mathrm{lo}},c)-\frac{\|\nabla D(\rho_{\mathrm{lo}},c)\|_2^2}{2m}
 &>1.5345\cdot10^{-39},\label{u:eq:lower-point}\\
 D(\rho_{\mathrm{hi}},c)&<-3.8249\cdot10^{-39},\label{u:eq:upper-point}\\
 D(\overline\rho,c)-0.037\|\nabla D(\overline\rho,c)\|_\infty
       +\frac m2(0.0119)^2
 &>7.0247\cdot10^{-7}.\label{u:eq:boundary-point}
\end{align}
\begin{proof}[Proof of \cref{u:thm:scalar}]
For an $m$-strongly convex function on a convex set,
\begin{equation}\label{u:eq:sc-bound}
 D(\rho,z)\ge D(\rho,c)+\nabla D(\rho,c)\cdot(z-c)
                  +\frac m2\|z-c\|_2^2.
\end{equation}
Completing the square and using \cref{u:eq:lower-point} proves $D(\rho_{\mathrm{lo}},z)>0$ throughout $K$. Every $z\in\partial K$ has $\|z-c\|_2\ge0.012>0.0119$ and $\|z-c\|_1\le0.036<0.037$. Hence \cref{u:eq:boundary-point} proves $D(\overline\rho,z)>0$ on $\partial K$. \cref{u:lem:exterior,u:lem:cover-logic} now prove strict positivity of $D(\overline\rho,\cdot)$ everywhere outside $\operatorname{int}K$.

Since $D(\rho,z)=(A+I)(R_{\mathrm U}(z)-\rho)$ and $A+I>0$, all points outside $\operatorname{int}K$ have $R_{\mathrm U}>\overline\rho$. \cref{u:eq:upper-point} provides a point inside $K$ with $R_{\mathrm U}<\rho_{\mathrm{hi}}<\overline\rho$. Therefore $R_{\mathrm U}$ has a global minimizer in $\operatorname{int}K$. The lower point bound, together with exterior positivity and monotonic decrease of $D$ in $\rho$, gives
\[
 \rho_{\mathrm{lo}}<\min R_{\mathrm U}<\rho_{\mathrm{hi}}.
\]
At its value $\rhostar$, the function $D(\rhostar,\cdot)$ is nonnegative and has zero gradient at the minimizing point. Strong convexity \cref{u:eq:strong-convexity} makes that point unique in $K$, and exterior positivity excludes any other global minimizer.

For localization, strong convexity also gives
\[
 \|z_*-c\|_2\le\frac{\|\nabla D(\rhostar,c)\|_2}{m}.
\]
The gradient is affine in $\rho$, so its norm is bounded by coordinatewise endpoint bounds at $\rho_{\mathrm{lo}},\rho_{\mathrm{hi}}$. The interval calculation gives
\begin{equation}\label{u:eq:location}
 \|z_*-c\|_2<1.346\cdot10^{-36}<10^{-30}.
\end{equation}
This places the minimizer inside \cref{u:eq:smallbox}.

Finally, a second solution of \cref{u:eq:stationary} in that box is impossible. Its zero gradient would, by strict convexity, make it the unique minimum of its $D(\rho,\cdot)$ on $K$. If its $\rho$ exceeded $\rhostar$, then $D(\rho,z_*)<0$, contradicting that minimum being zero. If its $\rho$ were smaller, its zero would give a ratio below the global minimum. At equal $\rho$, strict convexity gives the same point. \end{proof}

The proof uses explicit derivatives of the layer integrals. We collect
them below to complete the specification of the local checks.

\subsection{Derivative Formulas}\label{u:app:derivatives}
For completeness, we collect the derivatives used in the preceding
unrestricted certificate. They follow by differentiating the layer
integrals, including their moving support endpoints.

All derivatives in this subsection are with respect to $(a,p,r)$, and $x_0=p^{1/a}$, $y_0=r^{1/a}$, $k=1-a$. Write
\[
 x_a=-\frac{x_0\log x_0}{a},\qquad
 y_a=-\frac{y_0\log y_0}{a},\qquad
 h_x=(1-x_0)^{-a},\quad h_y=(1-y_0)^{-a}.
\]
The formulas hold in the open region $x_0+y_0<1$, which contains all local verification boxes. With $T_m$ and $J_m$ as in \cref{u:eq:Tm,u:eq:Jm}, the first derivatives are
\begin{align*}
 A_a&=-pT_1(a,1-x_0),&
 A_p&=-\frac{1-x_0^k}{k},& A_r&=0,\\
 I_a&=pT_1(a,y_0)+rT_1(a,x_0)+prJ_1,\\
 I_p&=T_0(a,y_0)+rJ_0,&
 I_r&=T_0(a,x_0)+pJ_0.
\end{align*}
The second derivatives of the seller mean are
\begin{align*}
 A_{aa}&=-pT_2(a,1-x_0)+\frac{x_0(\log x_0)^2}{a},\\
 A_{ap}&=-T_1(a,1-x_0)+\frac{x_a}{p},&
 A_{pp}&=\frac{x_0}{ap^2},
\end{align*}
with all derivatives involving $r$ equal to zero. The derivatives of first-best gains are
\begin{align*}
 I_{aa}&=pT_2(a,y_0)+rT_2(a,x_0)+prJ_2
 -\frac{p y_0(\log y_0)^2}{a}h_y
 -\frac{r x_0(\log x_0)^2}{a}h_x,\\
 I_{ap}&=T_1(a,y_0)+rJ_1-\frac{r x_a}{p}h_x,\\
 I_{ar}&=T_1(a,x_0)+pJ_1-\frac{p y_a}{r}h_y,\\
 I_{pp}&=-\frac{r x_0}{ap^2}h_x,\\
 I_{rr}&=-\frac{p y_0}{ar^2}h_y,\\
 I_{pr}&=J_0.
\end{align*}
To verify the formulas, start from
\[
 \revise{A}{A_a(p)}=1-x_0-pT_0(a,1-x_0),\qquad
 \revise{I}{I_a(p,r)}=pT_0(a,y_0)+rT_0(a,x_0)+prJ_0.
\]
Differentiating with respect to $p$ or $r$, the moving-cutoff terms cancel because $p x_0^{-a}=r y_0^{-a}=1$. Differentiation in $a$ has the same cancellation at first order. Differentiating again leaves the displayed moving-endpoint terms. For example,
\[
 \partial_a I_p=T_1(a,y_0)+rJ_1
     -r[x_0(1-x_0)]^{-a}x_a
 =T_1(a,y_0)+rJ_1-\frac{r x_a}{p}h_x.
\]
The other formulas follow identically, using
$\partial_p x_0=x_0/(ap)$ and $\partial_r y_0=y_0/(ar)$.

For $C_\rho=1-\rho$, the Hessian of $D$ is therefore
\begin{align*}
 D_{aa}&=C_\rho A_{aa}+2pr/k^3-\rho I_{aa},\\
 D_{ap}&=C_\rho A_{ap}+r/k^2-\rho I_{ap},\\
 D_{ar}&=p/k^2-\rho I_{ar},\\
 D_{pp}&=\frac{x_0}{ap^2}(C_\rho+\rho r h_x),\\
 D_{rr}&=\frac{\rho p y_0}{ar^2}h_y,\\
 D_{pr}&=1/k-\rho J_0.
\end{align*}
These are the formulas evaluated by \texttt{verify\_local.py}.

This completes the unrestricted certificate. The MHR problem uses
the same global-to-local logic, with residual and flow enclosures
in place of explicit power-law integrals.

\section{\texorpdfstring{Proof of Theorem~\ref*{thm:certification}}{Proof of the MHR welfare theorem}}\label{app:mhr}
This appendix completes the evaluation of \eqref{eq:w-main}. Multiplier
bounds and residual identities reduce the global problem to a compact
domain. We then prove that the interval certificate gives the sharp
ratio, and justify each of its acceptance tests.

\subsection{Multiplier Bounds}\label{sec:gft}
The global welfare argument must exclude $\revise{\lambda}{\alpha}\downarrow0$ and
$\revise{\lambda}{\alpha}\to\infty$. The following fixed-multiplier gains bound, also
given by \cite[Theorem~4.1]{LQW26}, is sufficient. We record that
result in our parameterization and derive only the elementary estimate
used by the welfare certificate.

For $k=1-a$, define
\[
 M_k(z)=k\int_0^1t^{k-1}e^{zt}\dd t.
\]
Let $\zeta(a)$ be the unique solution in $(a,2a)$ of
\begin{equation}\label{eq:zeta}
 \frac1{M_k(\zeta)}=\frac{\zeta+1-2a}{k}e^{-\zeta}.
\end{equation}

\begin{proposition}[An endpoint lower bound]\label{prop:gft}
For every $0<a<1$, all bounded MHR-buyer instances satisfy $V_{\revise{\lambda}{\alpha}}\ge c(a)G$, where
\[
 c(a)=1/M_k(\zeta(a)).
\]
In particular
\begin{equation}\label{eq:endpoint-bound}
 c(a)\ge\left[1+\frac{k}{1+k}(e^{2a}-1)\right]^{-1}.
\end{equation}
\end{proposition}
\begin{proof}[Proof of \cref{prop:gft}]
The universal bound and the existence and uniqueness of $\zeta$ are
\cite[Theorem~4.1 and Section~4.1]{LQW26}, with \revise{their multiplier
$\alpha=\lambda$}{the same multiplier $\alpha$} and parameter $t=k=1-a$. For the additional estimate,
convexity gives $e^{zt}\le1+t(e^z-1)$ for $t\in[0,1]$. Hence
\[
 M_k(\zeta)\le1+\frac{k}{k+1}(e^\zeta-1)
 \le1+\frac{k}{k+1}(e^{2a}-1),
\]
which proves \eqref{eq:endpoint-bound}.
\end{proof}

This estimate excludes the extreme multipliers in
\cref{lem:finite-domain}. We next derive the welfare-specific identities
that control the remaining parameters.

\subsection{Residual Identities and Endpoint Variations}\label{sec:variations}
The MHR residual gives an exact seller optimum at its zero. To certify
the global minimum, we also need its meaning away from a zero and
its variation with the buyer endpoints.

The matching seller \eqref{eq:hard-seller} is meaningful whenever $z<L$, even when $\mathcal R\ne0$. Its optimized dual value, mean and efficient gains are denoted $V,A,G$. All derivatives in this subsection hold the displayed buyer parameters fixed unless a direction is explicitly specified.

\begin{proposition}[Exact residual identity]\label{prop:offroot}
For every such trial ratio,
\begin{equation}\label{eq:offroot}
 (1-\rho)A+V-\rho G=\frac{d^a}{k}\mathcal R,
 \qquad \mathcal R_\rho=-\frac{k(A+G)}{d^a}<0.
\end{equation}
Consequently
\begin{equation}\label{eq:newton}
 \rho-\frac{\mathcal R}{\mathcal R_\rho}=\frac{A+V}{A+G}.
\end{equation}
Thus a Newton update is an explicit seller's dual welfare ratio, not merely a numerical root-finding operation.
\end{proposition}
\begin{proof}[Proof of \cref{prop:offroot}]
On the head put $p(s)=e^{\ell-a-y(s)}$. Equation \eqref{eq:head} gives
\[
 wp'=k(p-b),\qquad F'=aF/w,
 \qquad (Fpw)'=-Fp'-kFb.
\]
The buyer score stop-loss at seller score $y$ is $p-kq$ on the head, and zero at the tail score $U$. Integrating the head, including the atom at zero, therefore gives
\[
 V=\frac{a}{k}qF(z)-\frac1k\int_0^zFp'\dd s
   =\int_0^zFb\dd s+\frac{qF(z)W}{k}.
\]
Also $A=L-\int_0^LF$ and $G=H(L)+\int_0^LFQ$, so
$\kappa A-\rho G=-\int_0^LFb$. Subtract the tail and use
$F(s)=d^a(U-s)^{-a}$ there. This proves the first identity.

For the derivative use $p=e^{\ell-a-y}$ and the integrating factor
$P=F(b-p)$ along the score flow. If $f=(y-s)/D$, direct differentiation gives
\[
 P_y+Pf_s=0,\qquad Pf_\rho=Ff(1-Q).
\]
The initial cost is fixed, hence
\begin{equation}\label{eq:adjoint-rho}
 F(z)(b(z)-q)z_\rho=\int_0^zF(s)(1-\revise{Q(s)}{\overline F_B(s)})\dd s.
\end{equation}
Implicit differentiation of the endpoint gives $b(L)L_\rho=L+H(L)$.
Differentiating \eqref{eq:R}, multiplying by $d^a/k$, and using \eqref{eq:adjoint-rho} yields
\[
 \frac{d^a}{k}\mathcal R_\rho=
 \int_0^LF(s)(1-\revise{Q(s)}{\overline F_B(s)})\dd s-L-H(L)=-(A+G).
\]
The first identity then rearranges to \eqref{eq:newton}.
\end{proof}

There are two useful normalizations, which must not be confused away from a zero:
\begin{equation}\label{eq:potentials}
 \mathcal F=\mathcal R/q,\qquad \Psi=\mathcal R/[qW^k].
\end{equation}
Both have the same zeros and signs as $\mathcal R$. We exclude stationary points of $\mathcal F$ globally and certify strict convexity of $\Psi$ locally. At a zero their stationary conditions agree, but an arbitrary signed derivative of one is not an off-root derivative test for the other.

\begin{proposition}[Two algebraic directional tests]\label{prop:endpoint-tests}
Let $p_0=\revise{F_\rho(0)}{F_{S,\rho}(0)}$, the hard seller's zero-cost atom. Translating both buyer endpoints, that is, differentiating in $\ell$ at fixed $T=U-\ell$, gives
\begin{equation}\label{eq:floor-balance}
 \mathcal F_\ell\big|_T=\frac{k}{qd^a}\{\kappa-p_0(1-e^{\ell-a})\}.
\end{equation}
When $\ell<z<L<U$, the cap derivative is
\begin{align}
 \mathcal F_U&=W^{-a}\mathcal H,\label{eq:cap-test}\\
 \mathcal H&=W+1-2a-k\rho e^W
 +\frac{\kappa}{q}\left[\left(\frac Wd\right)^a\{d+k(L+1)\}-W-k\right].\notag
\end{align}
\end{proposition}
\begin{proof}[Proof of \cref{prop:endpoint-tests}]
In the translation direction $q$ is fixed. Put $c(s)=b(s)+b'(s)$, equal to $1$ below $\ell$ and $\kappa$ above it. The same integrating factor as before gives, at fixed score,
\[
 P(s)s_\ell\big|_y=\int_0^sFc\dd t.
\]
Since the terminal score changes at unit speed,
\[
 P(z)z_\ell=\int_0^zFc\dd s+F(z)qw(z)/k.
\]
The identity $[F(b-p)-Fpw/k]'=Fc$, with initial value $p_0(1-e^{\ell-a})$, implies
$P(z)(z_\ell-1)=-p_0(1-e^{\ell-a})$.
Further, $b(L)L_\ell=\rho \revise{Q(L)}{\overline F_B(L)}$, so $b(L)(1-L_\ell)=\kappa$.
Differentiating $\mathcal R$ and integrating the tail derivative by parts gives
\[
 \frac{d^a}{k}\mathcal R_\ell\big|_T=P(z)(z_\ell-1)+\kappa,
\]
which proves \eqref{eq:floor-balance}.

For the cap, $z_U=(W-a)/D(h,z)$ and $b(L)L_U=\rho q$. Substitution into the differentiated residual cancels the head sensitivity:
\[
 \mathcal R_U=q(1-2a)W^{-a}-k\rho qd^{-a}
      +ka\int_z^Lb(s)(U-s)^{-a-1}\dd s.
\]
Since $q_U=-q$, $\mathcal F_U=(\mathcal R_U+\mathcal R)/q$. On the stated phase $b'=-b+\kappa$, hence
\[
 a\int_z^Lb(U-s)^{-a-1}\dd s-\int_z^Lb(U-s)^{-a}\dd s
 =[b(U-s)^{-a}]_z^L-\frac{\kappa}{k}(W^k-d^k).
\]
Finally $b(L)-\rho q=\kappa(L+1)$ and $b(z)/q=\kappa/q+\rho e^W$ reduce this expression to \eqref{eq:cap-test}.
\end{proof}

The deficit identity and the two directional tests supply the
analytic inputs for the global exclusions in the next subsection.

\subsection{Global Certification}\label{sec:global}
We now prove that the finite computation stated below closes the entire parameter domain, rather than only a neighborhood of a numerical stationary point.

\begin{lemma}[Finite-domain reduction]\label{lem:finite-domain}
Let $\bar\rho=911390/10^6$. Write $\fl(x)$ for the nearest IEEE binary64 value to $x$. Every proposed universal lower bound below $\bar\rho$ is valid if it is valid on
\begin{equation}\label{eq:C}
 \mathcal C=\{a_-\le a\le a_+,\ 0\le\ell\le a\le U\le29\},
 \qquad a_-=\fl(.0875),\quad a_+=\fl(.985).
\end{equation}
 boundary values of $\revise{r}{R_{\mathrm M}}$ are limits from the open \revise{canonical}{reduced parameter} domain. The profile has a continuous extension to this compact set.
\end{lemma}
\begin{proof}[Proof of \cref{lem:finite-domain}]
By \cref{prop:gft}, welfare ratios at a fixed multiplier are at least $c(a)$, because $c(a)\le1$. Formula \eqref{eq:endpoint-bound} gives
\[
 c(a)\ge\frac2{1+e^{2a_-}}>.91272>\bar\rho\quad(a\le a_-),
\]
and
\[
 c(a)\ge[1+(1-a_+)(e^2-1)]^{-1}>.91254>\bar\rho
 \quad(a\ge a_+).
\]
The coarse displayed decimals are conservative summaries. Only the comparisons with $\bar\rho$ are needed; the verifier checks these directly with outward arithmetic.

At a root $\revise{r}{R_{\mathrm M}(a,\ell,U)}\le\bar\rho$, the matching seller has $L<2.5$. Indeed $\ell\le a_+<2.5$, and
\[
 (1-\bar\rho)2.5-\bar\rho e^{a_+-2.5}>.021>0
\]
puts the unique endpoint below $2.5$. Its head mass satisfies
\[
 F(z)=\left(\frac{U-L}{U-z}\right)^a\ge1-L/U.
\]
If $U\ge29$, then $F(z)>1-2.5/29>\bar\rho$. Keep this seller and the multiplier fixed and lower only the buyer cap. On the head, the left derivative of $V_{\revise{\lambda}{\alpha}}$ is $q$ per unit seller mass; the tail has score exactly the old $U$ and contributes zero on either side of this downward change. Thus
\[
 \partial_U^-V_{\revise{\lambda}{\alpha}}=qF(z),\quad \partial_UG=q,\quad
 \partial_U^-\frac{A+V_{\revise{\lambda}{\alpha}}}{A+G}
 =\frac q{A+G}(F(z)-\revise{r}{R_{\mathrm M}(a,\ell,U)})>0.
\]
Lowering the cap strictly improves this witness; optimizing the seller can only improve further. On any interval $[29,U_0]$ containing a value below $\bar\rho$, a minimum must therefore occur at $29$.

For completeness, continuity at the finite parameter boundaries need not be deduced from a singular head equation. On a compact parameter region with $a$ bounded away from zero and one and $U\ge a_->0$, every interior minimizing seller from \eqref{eq:hard-seller} has $S\le U$ and normalized score $v\le U$. Changing $a$ to $a'$ changes its score to
$s+(a'/a)(v-s)$, \revise{uniformly continuously}{continuously, uniformly over the compact parameter region}. The \revise{canonical buyer}{capped-exponential buyer} scores, coupled by the same exponential variable, are uniformly continuous in $L^1$ as $(a,\ell,U)$ varies: the only moving jump is at the cap and has vanishing probability mass under a small parameter change. Buyer values have the same property. The supremum of the linear allocation functionals is Lipschitz in the $L^1$ norm of their coefficients, even if the changed multiplier requires ironing anew. The feasible monotone quantile allocation class can be kept fixed: averaging allocations within a buyer atom preserves interims of other agents and feasibility, and does not change the objective because the buyer coefficient is constant on that atom. Hence the score bounds just obtained, together with
\[
 A+G\ge\E B\ge1-e^{-a_-}>0
\]
give a uniform modulus for each matching witness's ratio. Using the optimizer at either of two interior triples as a trial seller at the other proves the same modulus for $\revise{r}{R_{\mathrm M}}$. It extends continuously to the finite closure. Bounds on these limiting buyers follow against any fixed finite seller by continuity of its Lagrangian, and then by seller approximation.

This justifies the compact-interval minimum used above and handles $\ell=a$, $U=a$ and $\ell=0$ without assuming an unproved differentiable boundary optimizer. Infinite-cap components are limits of finite caps for each fixed seller: the \revise{canonical buyer}{capped-exponential buyer} values and scores converge in $L^1$. Universal fixed-multiplier inequalities therefore pass to them. This completes the domain reduction.
\end{proof}

The relevant interior region and ratio interval are
\begin{align}
 \mathcal B&=[\fl(.19),\fl(.39)]\times[\fl(.002),\fl(.16)]
                 \times[\fl(.85),\fl(1.65)],\label{eq:B}\\
 I_0&=[.9113893680,.9113893683].\label{eq:I0}
\end{align}
Ratio endpoints in this manuscript and the verifier are exact rationals, not rounded decimal inputs. Define $c$ and $\mathcal N$ by applying $\fl$ to every coordinate below:
\begin{align}
 c&=(.28523064649549684,.04945697286361709,1.1746687514380845),\label{eq:center}\\
 \mathcal N_a&=[.28521064649549682,.28525064649549686],\notag\\
 \mathcal N_\ell&=[.04943697286361709,.049476972863617089],\label{eq:N}\\
 \mathcal N_U&=[1.1745887514380844,1.1747487514380845].\notag
\end{align}

\begin{proposition}[Finite certificate assertions]\label{prop:machine}
The elementary interval algorithms specified in \cref{sec:verification} and supplied with their complete inputs in the accompanying certificate repository prove the following assertions.
\begin{enumerate}[label=(\alph*)]
\item At target $\bar\rho$, every open-domain point of $\mathcal C\setminus\mathcal B$ has positive residual; boundary limits have ratio at least $\bar\rho$.
\item Uniformly for $\rho\in I_0$, every nonpositive stationary point of $\mathcal F$ in $\mathcal B$ lies in $\mathcal N$.
\item For all $\theta\in\mathcal N$ and $\rho\in[.9113893679,.9113893684]$, with $D_0=\diag(1,1,10)$,
\begin{equation}\label{eq:strong-convex}
 D_0\nabla_\theta^2\Psi\,D_0\succeq m\revise{I}{I_3},\qquad m>.59908114.
\end{equation}
\item The local and center signs are
\begin{align*}
 \min_{\mathcal N}\Psi(\theta;.9113893680)&>3.8564462\cdot10^{-10},\\
 \Psi(c;.9113893683)&<-5.3617976\cdot10^{-10},\\
 \min_{\mathcal N}\Psi(\theta;.911389368124)&>2.9087841\cdot10^{-12},\\
 \Psi(c;.911389368127)&<-2.1984636\cdot10^{-12}.
\end{align*}
\end{enumerate}
\end{proposition}

The finite assertions imply \cref{thm:certification}. We first
deduce that implication, then verify the assertions in
\cref{sec:verification}.

\begin{proof}[Proof of \cref{thm:certification}]
First take $\rho_0=.9113893680$. If $\mathcal F$ is negative somewhere in $\mathcal B$, its negative minimum is interior: (a), continuity and the sign characterization give nonnegativity on $\partial\mathcal B$. The function is continuously differentiable there. Indeed the score vector field is continuous across $s=\ell$, the crossing is transverse, and its one-sided variational solutions match with no jump; the stop-loss and endpoint equation are $C^1$. Moving an endpoint through the buyer floor also preserves the first derivatives of the integral in $\mathcal R$. Thus the negative minimum is stationary. Assertion (b) puts it in $\mathcal N$, contradicting (d). The finite-domain reduction proves a global ratio lower bound of $\rho_0$.

The second sign in (d) and the explicit seller of \cref{prop:offroot} give a global upper bound below $.9113893683$. By (a) and compactness there is a minimizer inside $\mathcal B$. At its ratio $\rho_{\mathrm M}$ the residual is nonnegative everywhere and zero at the minimizer, so its $\theta$ gradient vanishes. Assertion (b) therefore confines every minimizer to $\mathcal N$. The two refined signs in (d) now give \eqref{eq:sharp-bracket}.

At a zero the stationary conditions for $\mathcal R,\mathcal F,\Psi$ are equivalent. Strict convexity (c) gives at most one stationary point of $\Psi$ at each fixed $\rho$. There cannot be stationary zeros at two ratios $\rho_1<\rho_2$ in the specified interval: $\mathcal R_\rho<0$ throughout this local region, so a zero at $\rho_1$ becomes strictly negative at $\rho_2$, contradicting that a stationary zero of a strictly convex function is its minimum on $\mathcal N$. Existence has already been proved by global minimization. This proves \eqref{eq:stationary-system} and uniqueness. The matched seller's weak-duality upper bound equals the universal lower bound at this minimum, so this prior pair attains the second-best welfare ratio. Its normalized seller scores are bounded; therefore the same weak-* compactness argument used for budget duality applies to this continuous seller as well, and a budget-feasible optimal allocation exists.
\end{proof}

The global argument reduces the MHR claim to finitely many enclosing
tests. The next subsection proves the interval assertions used above.

\subsection{Interval Verification}\label{sec:verification}
This subsection supplies the mathematical acceptance rules for \cref{prop:machine}. All numerical sources and finite data accompany the paper in the certificate repository; \cref{sec:reproduction} gives the replay procedure. A program's search labels or floating-point proposals are never trusted signs.

\begin{proof}[Proof of \cref{prop:machine}]
We verify assertions (a)--(d) using the following outward enclosures.

\subsubsection{Outward Arithmetic}
Basic interval operations are the corresponding binary64 round-to-nearest operations enlarged by one representable number on both sides. Multiplication takes the four endpoint products; division rejects a denominator containing zero. IEEE subnormals and infinities use the same adjacency operation. Nonfinite final enclosures are rejected. Contraction and fast-math are disabled.

The elementary functions do not assume a system-library error bound. Put $\epsilon=2^{-53}$ and $l_2=\fl(.6931471805599453)$. A directed 200-bit MPFR calculation verifies $0<\log2-l_2<\epsilon/4$. To enclose $e^x$, choose an integer $n$, form $r=x-nl_2$ and check $|r|<.35$, then evaluate the degree-18 exponential Taylor polynomial by Horner's rule. Before scaling by $2^n$, enlarge its value by
\[
 \epsilon(100+2|n|).
\]
The range-reduction error is at most $\epsilon(1+|n|)$, the polynomial roundoff is less than $80\epsilon$ by the usual product bound $\gamma_j=j\epsilon/(1-j\epsilon)$, and the omitted tail is below $10^{-25}$. These fit strictly inside the allowance. Integer-exponent scaling is exact except for representational rounding, covered by outward adjacency.

For logarithms scale exactly to $x=2^ny$, $.707<y<1.415$, and set $t=(y-1)/(y+1)$, checking $|t|<.172$. Evaluate
\[
 2\sum_{j=0}^{12}\frac{t^{2j+1}}{2j+1}+nl_2
\]
with allowance $\epsilon(100+3|n|)$. The series tail is below $2\cdot10^{-22}$. Sterbenz's lemma makes $y-1$ exact; a $\gamma_{64}$ bound times an absolute sum below $.35$ covers the polynomial operations, perturbing $t$ costs less than $4\epsilon$, and the last product and addition cost less than $(2|n|+2)\epsilon$. Powers use these enclosing logarithm and exponential operations. The small-argument exprel series is enclosed by its next-tail exponential bound. Range checks and outward interval operations supply the remaining safeguards.

\subsubsection{Exterior Comparison}
In cost coordinates put $u=(1+w)e^{-w}$ and $\revise{\mathcal E(u)}{\chi(u)}=e^{-w}$. Equation \eqref{eq:head} becomes
\begin{equation}\label{eq:polygon-ode}
 u'=u-a\revise{\mathcal E(u)}{\chi(u)}-k e^{a-\ell+s}b(s),\qquad u(0)=1.
\end{equation}
One has $\revise{\mathcal E'(u)}{\chi'(u)}=1/w$ and $\revise{\mathcal E''(u)}{\chi''(u)}=e^w/w^3>0$. Secants through upper enclosures of this convex function give a polygon $\revise{\mathcal E_+}{\chi_+}\ge\revise{\mathcal E}{\chi}$. Replacing $\revise{\mathcal E}{\chi}$ by $\revise{\mathcal E_+}{\chi_+}$ gives a lower solution $\underline u\le u$. Comparison remains valid at $u=1$: the true vector field has one-sided Lipschitz constant at most one, since $\revise{\mathcal E}{\chi}$ is increasing.

The head ends at $u(s)=(1+U-a-s)e^{-(U-a-s)}$. The left curve decreases and the right curve increases; a lower head gives a lower bound $z_-$ on the intersection. To justify the first monotonicity, the cost equation is
\[
 ww'=k e^{a-\ell+s+w}b(s)-k-w.
\]
It starts with $w'>0$. At a first zero of $w'$, its derivative is positive because $b+b'\ge\kappa>0$; hence $w'$ cannot cross to negative values. Thus $(1+w)e^{-w}$ decreases. On a parameter box the forcing in \eqref{eq:polygon-ode} decreases with $a$ and $\ell$; use the smallest of these parameters in the forcing, the largest $a$ before $\revise{\mathcal E_+}{\chi_+}$, and the smallest cap distance. All coefficient roundings go in the lower-solution direction. Each polygon segment solves a linear ODE with constant and exponential forcing. Knot crossings are proposed numerically but accepted only after an enclosing one-sided flow inequality; downward jumps to the next knot preserve the subsolution.

The seller endpoint increases with both buyer endpoints, so solve its monotone equation conservatively at their upper values, obtaining $L_+$. The residual increases with its head join:
\[
 \partial_z\mathcal R=k[b(z)-q](U-z)^{-a}>0.
\]
Thus $z_-\ge L_+$ is immediately positive. Otherwise, if $L_+<U_l$, a box-wide lower bound is
\begin{align}\label{eq:box-bound}
 &e^{\ell_l-U_h}\min_{p\in\{1-a_l,1-a_h\}}(U_l-z_-)^p\notag\\[-2mm]
 &\quad-(1-a_l)\int_{z_-}^{L_+}
 [1-\rho+\rho\min(1,e^{\ell_h-s})]
 \max_{p\in\{a_l,a_h\}}(U_l-s)^{-p}\dd s.
\end{align}
Each branch of this integrand is log-convex and hence convex, separately below and above $\ell_h$; the maximum preserves convexity. A composite trapezoid rule split at $\ell_h$ is an upper integral bound. The program uses 32 subintervals per nonempty piece. A strictly positive outward-rounded version of \eqref{eq:box-bound} accepts the box; failure merely requires subdivision.

The exterior proof is a complete preorder binary tree on the initial box containing \eqref{eq:C}. At each nonterminal node the longest normalized side, with scales $(.2,.2,1)$ and a fixed tie rule, is bisected at its binary midpoint. A leaf is either verified by the preceding bound or wholly contained in $\mathcal B$. The replay checks the entire tree, validates every containment, and rejects missing or trailing bytes. At exact target $911390/10^6$ it verifies
\[
 47{,}273{,}387\text{ nodes},\qquad23{,}636{,}694\text{ leaves},\qquad0\text{ failures}.
\]
The four scalar bounds in \cref{lem:finite-domain} are checked by the same arithmetic before the tree is read. This proves assertion (a).

\subsubsection{Head Flow and Integral Derivatives}
The interior replay uses three complementary enclosures. Their acceptance rules are identical: a positive value, a nonzero derivative, or a valid stationary-point exclusion. An exception is never an exclusion.

The direct evaluator integrates \eqref{eq:head} in $y=(U-a)t$, $0\le t\le1$, together with its first variational equations. The physical invariant $0\le s\le y$ is imposed. A Picard tube $\mathcal T$ is verified from
\[
 s_0+[0,\Delta t]f([t,t+\Delta t],\mathcal T)\subseteq\mathcal T.
\]
It bounds a second-order Taylor step with a third-order remainder on smooth pieces. On steps crossing $s=\ell$, where the vector field is continuous but its derivatives change, it uses a first-order step with an enclosed second-order remainder, taking the hull of the two derivative branches. Variational equations have bounded measurable coefficients and no saltation jump because the vector field itself agrees at the crossing. For $v'=\revise{\alpha(t)}{c_1(t)}v+\revise{\beta(t)}{c_0(t)}$, where $c_1,c_0$ are generic interval-bounded coefficient functions, interval coefficient bounds give
\[
 v(t+\Delta t)\in e^{\Delta t\revise{[\alpha]}{[c_1]}}v(t)
   +\Delta t\,\operatorname{exprel}(\Delta t\revise{[\alpha]}{[c_1]})\revise{[\beta]}{[c_0]},
\]
by variation of constants. Exprel is increasing, so endpoint enclosures suffice. Intersecting with the parameter-center mean-value enclosure reduces wrapping.

The initial tube predictor uses $0\le s_y\le1$ on $\mathcal B$. Here is a \revise{uniform justification}{parameter-uniform proof} for that auxiliary bound. Since $w=y-s\ge0$ and $b\ge \revise{Q}{\overline F_B}$, both phases give
\[
 D-w\ge k(e^{w+a-\ell}-1)-w
 \ge 2a-\ell+\log(1-a)>0
\]
on $\mathcal B$; the last minimum is at $a=.19,\ell=.16$ and exceeds $.009$. The middle inequality is the unconstrained minimum over $w$. The subsequent Picard inclusion is checked independently, so the predictor is not itself an acceptance test.

The higher-order evaluator removes the floor crossing analytically. Before it, $b=1$ and the head is linear. With $r_0=a-\ell$ and $M(t)=(1-e^{-t})^{1/k}$, the score $y_0$ where cost reaches $\ell$ is the unique root of
\begin{equation}\label{eq:crossing}
 y_0-\ell-\int_0^{y_0}\frac{M(t+r_0)}{M(y_0+r_0)}\dd t=0.
\end{equation}
This follows by solving the linear ODE with integrating factor $M(y+r_0)$ and integrating by parts. Its derivative in $y_0$ is positive at the root since the cost derivative is positive. Verified bracketing and interval Newton enclose the root, and implicit differentiation encloses its first two parameter derivatives. A monotone-corner shortcut is used only when its sign condition $a+y_0-\ell<-\log a$ is verified; differentiating the integrand then proves monotonicity. Otherwise full-box bracketing is used.

The smooth post-floor flow is propagated with parameter jets through order two. A Taylor coefficient of a time series is obtained recursively from the ODE; a Picard tube for the entire jet system bounds the next coefficient and thus the remainder. Each step differentiates the flow with respect to an independent initial state and composes it with the preceding parameter jet by the chain rule. A fixed physical-score grid, with two short parameter-dependent endpoint bridges, avoids repeatedly transporting the cap parameter through the interior. A normalized-time frame is an alternative valid enclosure.

The endpoint $L$ is bracketed using $\kappa L-\rho H(L)$, whose derivative is $b(L)\ge\kappa$. It is then differentiated implicitly. Tail integrals and their parameter derivatives use composite Simpson enclosures. If $c_4$ encloses the fourth Taylor coefficient on an interval of width $h$, the integral minus its Simpson value lies in $-h^5c_4/120$. This follows from the fixed-sign Peano kernel and the ordinary Simpson remainder. Auxiliary midpoint integrals use the analogous second-derivative bound. All parameter boxes using the exponential tail derivatives first verify $\ell<z<L<U$ as appropriate; other boxes use the general exterior/value enclosures.

The scalar atom evaluator also integrates $J_y=a/D$ and uses the exact pre-floor integral to enclose
\[
 p_0=\left(\frac dW\right)^a
 \left(\frac{1-e^{\ell-a}}{1-e^{\ell-a-y_0}}\right)^{a/k}
 \exp\left(-\int_{y_0}^{U-a}\frac aD\dd y\right).
\]
It applies \eqref{eq:floor-balance} and \eqref{eq:cap-test}, with centered mean-value enclosures of the algebraic expression. The direct sensitivity evaluator can also compose the cap expression with the enclosed head sensitivities.

\subsubsection{Interior Coverage and Stationary Exclusions}
The interior proof contains 64,175 binary paths. Each path bisects the longest normalized side of \eqref{eq:B}, with scales $(.2,.158,.8)$, and the ratio interval is always \eqref{eq:I0}. The read-only audit verifies that the paths are distinct, prefix-free and satisfy the exact integer Kraft identity
\[
 \sum_{p}2^{d-|p|}=2^d,\qquad d=51.
\]
This proves complete coverage, including boundaries. There are 128,349 nodes in the corresponding full binary tree.

Each box is accepted only if an enclosing calculation proves $\mathcal F>0$, a gradient component is strictly signed, one of \eqref{eq:floor-balance}--\eqref{eq:cap-test} is strictly signed, or every stationary point is confined to $\mathcal N$. For the last alternative let $H$ enclose the Hessian and $g_c$ the center gradient. Every stationary point $c+\delta$ satisfies
\[
 \delta=-Cg_c+(I-CH)\delta
\]
for any fixed numerical matrix $C$. Iterative interval intersection either makes a coordinate empty, excluding every stationary point, or bounds their locations inside $\mathcal N$. No claim that $C$ is an exact inverse is needed. Containment of the entire box in $\mathcal N$ is also sufficient.

The method label in a saved leaf is only a hint. Independent replay may use a different valid enclosure, but must verify one of these same conclusions afresh. In particular it never accepts a negative value alone. Complete replay gives 64,175 checked leaves and zero failures. This proves (b); neither a partial tree nor a successful local calculation alone would suffice.

\subsubsection{Local Convexity and Sign Bounds}
For $\rho\in[.9113893679,.9113893684]$ the jet evaluator bounds the Hessian of $\Psi$ on \eqref{eq:N}. Coarser outward decimal summaries of the checked intervals are
\[
 \begin{pmatrix}
 [5.15202,5.29099]&[-.43729,-.31332]&[-.21640,-.17455]\\
 *&[1.56517,1.90290]&[-.05289,-.03354]\\
 *&*&[.04932,.07989]
 \end{pmatrix}.
\]
The full intervals, retained in the accompanying log, give the scaled Gershgorin lower bound $m>.59908114$ in \eqref{eq:strong-convex}. Positive interval $LDL^{\mathsf T}$ pivots provide an additional consistency check.

For any center $c\in\mathcal N$, integrating the Hessian along the segment to $\theta$ and minimizing the resulting quadratic gives
\begin{equation}\label{eq:convex-bound}
 \min_{\theta\in\mathcal N}\Psi(\theta;\rho)
 \ge\Psi(c;\rho)-\frac{\|D_0\nabla\Psi(c;\rho)\|^2}{2m}.
\end{equation}
The exact rational targets $.9113893680$ and $.911389368124$ give the two positive lower bounds in (d). Direct center evaluation at $.9113893683$ and $.911389368127$ gives the negative upper bounds there. All four checks are included in the local verifier; their numerical values in \cref{prop:machine} are rounded conservatively from its full enclosures. This proves (c)--(d).
\end{proof}

Both constants now have complete analytic acceptance rules. The final
appendix identifies the certificate inputs and the commands for replay.

\section{Certificate Data and Reproduction}\label{sec:reproduction}
The preceding appendices give the acceptance rules. Here we describe
their finite inputs and how to replay them.

The numerical part of the proof uses the accompanying certificate
repository, \texttt{sharp-welfare-certificates}, available at
\begin{center}\small\ttfamily
https://github.com/SolaQin/sharp-welfare-certificates
\end{center}
It contains all subdivision data, the verifier sources for both the
primary C++ and the optional Python interval implementations, and the
recorded acceptance runs. The logs record acceptance runs; the
mathematical proof objects are the finite input data together with the
independently enclosing tests. The manuscript and its bibliography
compile without downloading or executing this repository.

\subsection{Obtaining the Certificate Data}
Clone the repository or download it as an archive, then work from its
root directory. The finite inputs are committed directly, so no
preliminary extraction step is needed. The single entry point
\texttt{verify.py} compiles the verifiers, expands the compressed
exterior proof tree, and runs every check in both certificates.

The replay requires Python 3 with NumPy, a C++17 compiler, and MPFR/GMP.
The optional Python interval implementations additionally use
mpmath~\cite{mpmath}. Binary64 interval routines require IEEE
round-to-nearest, without fast-math or fused contraction.
Compiler include and library paths may be added for the local
installation.

\par\bigskip\noindent\begin{minipage}{\linewidth}
\subsection{Unrestricted-Prior Certificate}
From the repository root, change into \texttt{unrestricted} and run:
\begin{lstlisting}
python3 verify_geometry.py
python3 verify_export.py
c++ -O3 -std=c++17 -ffp-contract=off verify_cover_mpfr.cpp \
    -lmpfr -lgmp -o cover
c++ -O3 -std=c++17 -ffp-contract=off verify_local_point_mpfr.cpp \
    -lmpfr -lgmp -o local
./cover cover_boxes.bin 0
./local
\end{lstlisting}
\end{minipage}
The first program checks exact subdivision coverage. The second checks
that the binary MPFR input is a byte-for-byte export of that same
subdivision. The C++ programs check every nonlocal box and all exterior,
local Hessian, point, boundary, and localization inequalities.
Alternative Python interval implementations are also supplied.

\par\bigskip\noindent\begin{minipage}{\linewidth}
\subsection{MHR Certificate}
From the repository root, change into \texttt{mhr} and run:
\begin{lstlisting}
c++ -O3 -std=c++17 -ffp-contract=off -pthread \
    replay_residual_cover.cpp -lmpfr -lgmp -o replay
c++ -O3 -std=c++17 -ffp-contract=off -pthread \
    verify_candidate_neighborhood.cpp -lmpfr -lgmp -o local
c++ -O3 -std=c++17 -ffp-contract=off -pthread \
    verify_exterior_cover.cpp -lmpfr -lgmp -o exterior
python3 audit_interior.py interior.leaves
./replay interior.leaves 8
./local
gzip -dk exterior_cover.bin.gz
./exterior exterior_cover.bin .05 8
\end{lstlisting}
\end{minipage}
The final integer selects the number of worker threads and does not
change the acceptance tests. The exterior header uses little-endian
binary64; a different-endian implementation must decode it explicitly.
Subdivision paths and midpoint conventions are specified in the
sources. Every interval call must return a valid enclosure, and every
leaf must pass. The local program's
\texttt{PASS LOCAL CHECKS ONLY} refers only to the local component;
the global claim also needs both covers.

The proof is computer assisted. It does not rely on an optimizer's
termination, floating-point sign guesses, or a formal proof assistant.
All analytic reductions, acceptance inequalities, and remainder bounds
appear in the manuscript; the repository supplies their finite inputs
and implementations.

\subsection{Reproducing Every Certificate at Once}
The per-component commands above mirror the paper's proof structure.
To replay both certificates in a single run from the repository root,
use the bundled entry point, whose worker count affects only the MHR
checks:
\begin{lstlisting}
python3 verify.py --jobs 8
\end{lstlisting}
A successful run ends with
\texttt{FULL CERTIFICATE REPLAY PASSED: all.} and writes each program's
output to \texttt{logs/}. Passing \texttt{--check-only} instead checks
subdivision structure and binary export consistency without compiling
or evaluating any interval, and reports that the interval proofs have
not run.
\end{document}